\documentclass[sigconf]{acmart}
\usepackage{hyperref}
\usepackage{amsmath}
\usepackage{amssymb}
\usepackage{amsfonts}
\usepackage{amsthm}
\usepackage{bm}
\usepackage{enumitem}
\usepackage{multirow}
\usepackage{graphicx}
\usepackage{subfigure}
\usepackage{url}
\usepackage{balance}

\usepackage[linesnumbered,boxed,ruled,vlined]{algorithm2e}
\usepackage{thmtools}
\newcommand{\hide}[1]{} %hide
\newcommand{\vpara}[1]{\vspace{0.2cm}\noindent\textbf{#1 }}

\newcommand{\figref}[1]{Figure~\ref{#1}} %section reference
\newcommand{\beq}[1]{{\small \begin{equation}#1\end{equation}
}}

\newcommand{\besp}[1]{\begin{split}#1\end{split}}
\newcommand{\beal}[1]{{\small \begin{align}#1\end{align}}}

\DeclareMathOperator{\diag}{diag}
\DeclareMathOperator{\vol}{vol}

\DeclareMathOperator{\hlog}{\log^\circ}

\DeclareMathOperator{\ln1p}{log1p}
\DeclareMathOperator{\hln1p}{\ln1p^\circ}

\DeclareMathOperator{\tln}{trunc\_log}
\DeclareMathOperator{\htln}{trunc\_log^\circ}

\newcommand{\abs}[1]{\left\lvert#1\right\rvert}
\newcommand{\norm}[2]{\left\lVert#1\right\rVert_{#2}}

\newtheorem{theorem}{Theorem}
\newtheorem{lemma}{Lemma}
\newtheorem{definition}{Definition}

\newcommand{\etal}{{\em et al.}\ }

\newcommand\blfootnote[1]{%
  \begingroup
  \renewcommand\thefootnote{}\authornote{#1}%
  \addtocounter{footnote}{-1}%
  \endgroup
}

\title{NetSMF: Large-Scale Network Embedding as Sparse Matrix Factorization
}

\author{Jiezhong Qiu}
\email{qiujz16@mails.tsinghua.edu.cn}
\affiliation{%
  \institution{Tsinghua University}
}
\authornotemark[2]

\author{Yuxiao Dong}
\email{yuxdong@microsoft.com}
\affiliation{%
  \institution{Microsoft Research, Redmond}
}

\author{Hao Ma}
\authornote{Work performed while at Microsoft Research.}
\email{haom@fb.com}
\affiliation{%
  \institution{Facebook AI}
}

\author{Jian Li}
\email{lijian83@tsinghua.edu.cn}
\affiliation{%
  \institution{Tsinghua University}
}

\author{Chi Wang}
\email{wang.chi@microsoft.com}
\affiliation{%
  \institution{Microsoft Research, Redmond}
}

\author{Kuansan Wang}
\email{kuansanw@microsoft.com}
\affiliation{%
  \institution{Microsoft Research, Redmond}
}

\author{Jie Tang}
\email{jietang@tsinghua.edu.cn}
\affiliation{%
  \institution{Tsinghua University}
}
\additionalaffiliation{
\institution{Beijing National Research Center for Information Science and Technology~(BNRist)}
}

\hide{
\author{Jiezhong Qiu$^{\dag}$, Yuxiao Dong$^{\ddag}$, Hao Ma$^{\natural*}$, Jian Li$^{\sharp}$, Chi Wang$^{\ddag}$, Kuansan Wang$^{\ddag}$, and Jie Tang$^{\dag}$}

\blfootnote{
Work performed while at Microsoft Research.
}

\affiliation{$^{\dag}$Department of Computer Science and Technology, Tsinghua University}
\affiliation{$^{\dag}$Beijing National Research Center for Information Science and Technology, China}
\affiliation{$^{\ddag}$Microsoft Research, Redmond}
\affiliation{$^{\natural}$Facebook AI}
\affiliation{$^{\sharp}$Institute for Interdisciplinary Information Sciences, Tsinghua University}

\email{qiujz16@mails.tsinghua.edu.cn}
\email{{yuxdong, wang.chi, kuansanw}@microsoft.com, haom@fb.com, {lijian83, jietang}@tsinghua.edu.cn}
}

\renewcommand{\authors}{Jiezhong Qiu, Yuxiao Dong, Hao Ma, Jian Li, Chi Wang, Kuansan Wang, and Jie Tang}

\renewcommand{\shortauthors}{Qiu \etal}

\newcommand*\rot{\rotatebox{45}}
\copyrightyear{2019}
\acmYear{2019} 
\setcopyright{iw3c2w3}
\acmConference[WWW '19]{Proceedings of the 2019 World Wide Web Conference}{May 13--17, 2019}{San Francisco, CA, USA}
\acmBooktitle{Proceedings of the 2019 World Wide Web Conference (WWW '19), May 13--17, 2019, San Francisco, CA, USA}
\acmPrice{}
\acmDOI{10.1145/3308558.3313446}
\acmISBN{978-1-4503-6674-8/19/05}

\begin{document}
\begin{abstract}

%hide

%We study the problem of network embedding for large-scale networks. 
%Existing network embedding benchmark solutions, such as DeepWalk and LINE, have not achieved great efficiency and effectiveness together. 

We study the problem of large-scale network embedding, which aims to learn latent representations for network mining applications. 
Previous research shows that 1) popular network embedding benchmarks, such as DeepWalk, are in essence implicitly factorizing a matrix with a closed form, and 2)
the explicit factorization of such matrix generates more powerful embeddings than existing methods. 
However, directly constructing and factorizing this matrix---which is dense---is prohibitively expensive in terms of both time and space, making it not scalable for large networks. 

In this work, we present the algorithm of large-scale network embedding as sparse matrix factorization~(NetSMF). 
NetSMF leverages theories from spectral sparsification to efficiently sparsify the aforementioned dense matrix, enabling significantly improved efficiency in embedding learning. 
The sparsified matrix is spectrally close to the original dense one with a theoretically bounded approximation error, which helps maintain the representation power of the learned embeddings. 
We conduct experiments on networks of various scales and types. 
%Experiments on networks of various scales and types 
%Results show that for large networks with millions of vertices NetSMF achieves orders of magnitude speedup over existing dense matrix factorization approach. 
Results show that among both popular benchmarks %(i.e., DeepWalk and LINE) 
and factorization based methods, NetSMF is the only method that achieves both high efficiency and effectiveness. 
%Meanwhile, NetSMF also consistently and significantly outperforms DeepWalk---the common benchmark---in terms of both efficiency and effectiveness. 
We show that NetSMF requires only 24 hours to generate effective embeddings for a large-scale academic collaboration network with tens of millions of nodes, while it would cost DeepWalk months and is computationally infeasible for the dense matrix factorization solution. 
The source code of NetSMF is publicly available\footnote{\url{https://github.com/xptree/NetSMF}}. 

\end{abstract}


\hide{ % space is not enough, so i remove ccsxml and keywords
\begin{CCSXML}
<ccs2012>
<concept>
<concept_id>10002951.10003227.10003351</concept_id>
<concept_desc>Information systems~Data mining</concept_desc>
<concept_significance>500</concept_significance>
</concept>
<concept>
<concept_id>10002951.10003260.10003282.10003292</concept_id>
<concept_desc>Information systems~Social networks</concept_desc>
<concept_significance>500</concept_significance>
</concept>
<concept>
<concept_id>10003752.10003809.10003635.10010036</concept_id>
<concept_desc>Theory of computation~Sparsification and spanners</concept_desc>
<concept_significance>500</concept_significance>
</concept>
<concept>
<concept_id>10010147.10010257.10010293.10010309</concept_id>
<concept_desc>Computing methodologies~Factorization methods</concept_desc>
<concept_significance>500</concept_significance>
</concept>
</ccs2012>
\end{CCSXML}

\ccsdesc[500]{Information systems~Data mining}
\ccsdesc[500]{Information systems~Social networks}
\ccsdesc[500]{Theory of computation~Sparsification and spanners}
\ccsdesc[500]{Computing methodologies~Factorization methods}

\keywords{Representation Learning; Network Embedding; Spectral Graph Sparsification; Social Networks; Matrix Factorization; Randomized SVD}
}% space is not enough, so i remove ccsxml and keywords

\maketitle

\section{Introduction}

Recent years have witnessed the emergence of network embedding, which offers a revolutionary paradigm for modeling graphs and networks~\cite{hamilton2017representation}. 
The goal of network embedding is to automatically learn latent representations for objects in networks, such as vertices and edges. %, and groups. 
Significant lines of research have shown that the latent representations are capable of capturing the structural properties of networks, facilitating various downstream network applications, such as vertex classification and link prediction~\cite{perozzi2014deepwalk,tang2015line,grover2016node2vec,dong2017metapath2vec}.

%There have been several notable network embedding models emerged as 
Over the course of its development, the DeepWalk~\cite{perozzi2014deepwalk}, LINE~\cite{tang2015line}, and node2vec~\cite{grover2016node2vec} models have been commonly considered as powerful benchmark solutions for evaluating network embedding research. 
The advantage of LINE lies in its scalability for large-scale networks as it only models the first- and second-order proximities. 
%As a result, it is not able to capture structural relationships more than two hops. 
That is to say, its embeddings lose the
multi-hop dependencies in networks. 
%structural relationships more than two hops.
DeepWalk and node2vec, on the other hand, leverage random walks on graphs and % ~\cite{lovasz1993random} and 
skip-gram~\cite{mikolov2013efficient} with large context sizes to model nodes further away~(i.e., global structures).
Consequently, it is computationally more expensive for DeepWalk and node2vec to handle large-scale networks. 
For example, with the default parameter settings~\cite{perozzi2014deepwalk}, % and a machine with Intel Xeon E7 CPU (64 cores) and 1.7T memory, 
DeepWalk requires months to embed 
an academic collaboration network
%the Open Academic Graph (OAG)\footnote{\url{https://www.openacademic.ai/oag/}} 
of 67 million vertices and 895 million edges\footnote{{With the default DeepWalk parameters (walk length: 40 and $\#$walk per node: 80), 214+ billion nodes (67M$\times$40$\times$80) with a vocabulary size of 67 million are fed into skip-gram. As a reference, Mikolov et al. reported that training on Google News of 6 billion words and a vocabulary size of only 1 million cost 2.5 days with 125 CPU cores~\cite{mikolov2013efficient}.}}. The node2vec model, which performs high-order random walks, takes more time than DeepWalk to learn embeddings.

%Both LINE and DeepWalk use some hyperparameter to limit the computational cost. For example, DeepWalk limits the length of each random walk as a constant. 
%\yx{months }
% For example, with the default parameter settings~\cite{perozzi2014deepwalk} and a machine with Intel Xeon E7 CPU (64 cores) and 1.7T memory, DeepWalk requires months to embed the Open Academic Graph (OAG)\footnote{\url{https://www.openacademic.ai/oag/}} of 67 million nodes and 895 million edges\footnote{{With the default DeepWalk parameters (walk length: 40 and $\#$walk per node: 80), 214+ billion nodes (67M$\times$40$\times$80) with a vocabulary size of 67 million are fed into skip-gram. As a reference, Mikolov et al. reported that training on Google News of 6 billion words and a vocabulary size of only 1 million cost 2.5 days with 125 CPU cores~\cite{mikolov2013efficient}.}}. 
%In addition, both LINE and DeepWalk 
%The advantage of DeepWalk lines in its ab

% In addition, both LINE and DeepWalk learn network embeddings in an empirical way (with many hyper-parameters), leaving the learned representations without theoretical guarantees. 
%\jz{check the statement}
More recently, a study shows that both the DeepWalk and LINE methods can be viewed as implicit factorization of a closed-form matrix~\cite{qiu2018network}. %, whose closed form can be derived from the network~\cite{qiu2018network}. 
%The matrix varies with the hyperparameters used in the two methods, but it converges to a specific form as the hyperparameters go to extreme values, e.g., when the length of random walk is unconstrained. 
Building upon this theoretical foundation, the NetMF method was instead proposed to explicitly factorize this matrix, achieving more effective embeddings than DeepWalk and LINE. 
Unfortunately, it turns out that the matrix to be factorized is an $n\times n$ dense one with $n$ being the number of vertices in the network,  making it prohibitively expensive to directly construct and factorize for large-scale networks.

In light of these limitations of existing methods (See the summary in Table \ref{tbl:comparison}), we propose to study representation learning for large-scale networks with the goal of achieving efficiency, capturing 
global structural contexts, and having theoretical guarantees. 
Our idea is to find a sparse matrix that is spectrally close to the dense NetMF matrix implicitly factorized by DeepWalk.  
The sparsified matrix requires a lower cost for both construction and factorization. 
%Meanwhile, making it close to the original DeepWalk matrix will provide us %with powerful representation capacity. 
Meanwhile, making it spectrally close to the original NetMF matrix can guarantee that the spectral information of the network is maintained, and the embeddings learned from the sparse matrix is as powerful as those learned from the dense NetMF matrix. 
%Finally, our new approach can avoid many hyper-parameters. 

%In light of these issues of (distributed) DeepWalk and NetMF, we propose to study representation learning for large-scale networks. 
%The idea is to have a sparse matrix that is spectrally close to the dense NetMF matrix implicitly factorized by DeepWalk.  
%The sparsified matrix will enable both its fast construction and factorization. 
%Meanwhile, making it close to the original DeepWalk matrix will provide us %with powerful representation capacity. 
%Meanwhile, making it spetrally close to the original DeepWalk matrix can guarantee that the spectral information of the network is maintained, and the representation capability of the sparse matrix is as powerful. 
%Finally, our new approach can avoid many hyper-parameters.  

%what is netsmf and how
In this work, we present the solution to network embedding learning as sparse matrix factorization~(NetSMF). 
NetSMF comprises three steps. 
First, it leverages the spectral graph sparsification technique~\cite{cheng2015spectral,cheng2015efficient} to find a sparsifier for a network's random-walk matrix-polynomial.
%First, it leverages the random-walk matrix polynomial sparsification technique~\cite{cheng2015spectral,cheng2015efficient} to find a sparsifier for a network's graph Laplacian matrix.
Second, it uses this  sparsifier to construct a matrix with significantly fewer non-zeros than, but spectrally close to, the original NetMF matrix. 
Finally, it performs randomized singular value decomposition to efficiently factorize the sparsified NetSMF matrix, yielding the embeddings for the network. 

% This table follows the style of KDD'16 best paper
% See Table 1 in https://www.kdd.org/kdd2016/papers/files/rfp0110-hooiA.pdf
\begin{table}[t!]
\caption{The comparison between NetSMF and other popular network embedding algorithms.}
\label{tbl:comparison}
\centering \small
\begin{tabular}[htbp]{c|cccc|c}
\hline \hline
 & \rot{LINE} & \rot{DeepWalk}  & \rot{node2vec} & \rot{NetMF} & \rot{NetSMF}  \\ \hline
Efficiency & $\surd$& & & &$\surd$ \\
Global context & & $\surd$&$\surd$ &$\surd$& $\surd$ \\
Theoretical guarantee & & & & $\surd$ & $\surd$ \\
High-order proximity & & &$\surd$ & & \\
\hline \hline
\end{tabular}
\normalsize
\end{table}	

With this design, NetSMF offers both efficiency and effectiveness with guarantees, as the approximation error of the sparsified matrix is theoretically bounded. 
We conduct experiments in five networks, which are representative of different scales and types. 
% The largest one---Open Academic Graph (OAG)---consists of 67 million vertices and 895 million edges. 
Experimental results show that for million-scale or larger networks, NetSMF achieves orders of magnitude speedup over NetMF, while maintaining competitive performance for the vertex classification task. 
In other words, both NetSMF and NetMF outperform well-recognized network embedding benchmarks (i.e., DeepWalk, LINE, and node2vec),  but NetSMF addresses the computation challenge faced by NetMF. % the only model reach the balance between effectiveness and efficiency. 

%When compared to popular network embedding baselines, such as DeepWalk and LINE, NetSMF 

%NetSMF also outperforms DeepWalk in both efficiency and effectiveness. 
%\yx{to have the following statement or not}Note that though LINE is faster than NetSMF, it is at the cost of losing structural information more than two hops away, leaving its embedding power consistently and significantly less competitive than NetSMF, NetMF, and DeepWalk. 
% In other words, NetSMF achieves the balance between effectiveness and efficiency. 

To summarize, we introduce the idea of network embedding as sparse matrix factorization and present the NetSMF algorithm, which makes the following contributions to network embedding:

\vpara{Efficiency.} NetSMF reaches significantly lower time and space complexity than NetMF. Remarkably, NetSMF is able to generate embeddings for a large-scale academic network of 67 million vertices and 895 million edges on a single server in 24 hours, while it would cost months for DeepWalk and node2vec, and is computationally infeasible for NetMF on the same hardware. 

\vpara{Effectiveness.} NetSMF is capable of learning embeddings that maintain the same representation power as the dense matrix factorization solution, making it consistently outperform DeepWalk and node2vec by up to 34\% and LINE by up to 100\% for the multi-label vertex classification task in networks.  
    
\vpara{Theoretical Guarantee.} NetSMF's efficiency and effectiveness are theoretically backed up. The sparse NetSMF matrix is spectrally close to the exact NetMF matrix, and the approximation error can be bounded, maintaining the representation power of its sparsely learned embeddings. 
    
    % \item \textbf{Hyperparameter free} \yx{can we make this claim?} With theoretical guarantees, NetSMF eliminates various hyperparameters in DeepWalk, including the walk length and $\#$walk starting from per node, whose settings require grid search empirically. In addition, NetSMF can determine embeddings' dimensionality by using the criterion of e.g., captured energy---that is, letting the chosen singular values capture a high percentage (e.g., 90\%) of the total energy\jz{ref}. 

\hide{
\vpara{Organization} The rest of the paper is organized as follows.
Section~\ref{sec:pre} reviews necessary preliminaries.
Section~\ref{sec:netsmf} introduces the proposed NetSMF algorithm in detail.
In Section~\ref{sec:exp}, we conduct extensive experiments and parameter analysis.
Finally, Section~\ref{sec:related} summarizes related work and Section~\ref{sec:conclusion} concludes this work with future directions looked into. 
}
\section{Preliminaries}
\label{sec:pre}

\begin{table}[t!]
\centering
\caption{Notations.}
\label{tbl:notation}
\centering \small
\begin{tabular}[htbp]{c|l}
\hline \hline
\textbf{Notation} & \textbf{Description} \\ \hline
$G$ &  input network\\\hline
$V$ &  vertex set of $G$  with $|V|$=n\\\hline
$E$ &  edge set of $G$ with $|E|=m$\\\hline
$\bm{A}$ & adjacency matrix of $G$\\\hline
$\bm{D}$ & degree matrix of $G$ \\\hline
$\vol{(G)}$ &  volume of $G$\\\hline
$b$ &  number of negative samples \\\hline
$T$ &  context window size\\\hline
$d$ &  embedding dimension \\\hline
$\bm{L}$ &  random-walk molynomial of $G$~(Eq.~\eqref{eq:L})\\\hline
$\widetilde{\bm{L}}$ & $\bm{L}$'s sparsifier \\ \hline
$\bm{M}$ & $\frac{1}{T}\sum_{r=1}^T (\bm{D}^{-1}\bm{A})^r \bm{D}^{-1}$ \\\hline
$\widetilde{\bm{M}}$ & $\bm{M}$'s sparsifier \\\hline
$\htln\left(\frac{\vol(G)}{b}\bm{M}\right)$ &  NetMF matrix \\ \hline
$\htln\left(\frac{\vol(G)}{b}\widetilde{\bm{M}}\right)$ &  NetMF matrix sparisifier \\ \hline
$M$ &  number of non-zeros in $\widetilde{\bm{L}}$  \\ \hline
$\epsilon$ & approximation factor \\
\hline
$[x]$ & set $\{1, 2, \cdots, x\}$ for  positive integer $x$\\
 \hline\hline
\end{tabular}
\normalsize
\end{table}

%\subsection{Problem Formulation and Notations}

%\vpara{Network Embedding:}
Commonly, the problem of network embedding is formalized as follows:  
Given an undirected and weighted network $G=(V, E, \bm{A})$ with $V$ as the vertex set of $n$ vertices, $E$ as the edge set of $m$ edges, and $\bm{A}$ as the adjacency matrix, the goal is to learn a function $V \rightarrow \mathbb{R}^{d}$ that maps each vertex to a $d$-dimensional~($d\ll n$) vector  that captures  its structural properties, e.g., community structures. The vector representation of each vertex can be fed into downstream applications such as link prediction and vertex classification.

\hide{
\vpara{Notations}In this work, we denote $d_i=\sum_{j}\bm{A}_{ij}$ to be the generalized degree of the $i$-th vertex,  $\bm{D} = \diag{(d_1, \cdots, d_{n})}$ to be the degree matrix, and $\vol{(G)}=\sum_i \sum_j \bm{A}_{ij}$ to be the volume of the graph. Moreover, $T$ and $b$ are the context window size and the number of negative sampling in skip-gram model, respectively. We finally introduce the definition of element-wise matrix function

\begin{definition} Let $f: \mathbb{R} \rightarrow \mathbb{R}$ be a real variable function and $\bm{A}\in\mathbb{R}^{n\times m}$ be a real-valued matrix.
The element-wise matrix function
$f^\circ: \mathbb{R}^{n\times m} \rightarrow \mathbb{R}^{n\times m}$ is defined  by $f^\circ(\bm{A}) = \left[f\left(\bm{A}_{ij}\right)\right]_{n\times m}$.
\end{definition}
For example, $\hlog(\bm{A})$ represents the element-wise matrix logarithm~\cite{hom1991topics}. We highlight this notation here 
to differentiate it from
matrix logarithm, $\log{(\bm{A})}$, which has special meaning in mathematics~\cite{wiki:matrixlog}. 
}%end of hide

% \subsection{Revisit DeepWalk as Matrix Factorization}

One of the pioneering work on network embedding is the DeepWalk model ~\cite{perozzi2014deepwalk}, which has been consistently considered as a powerful benchmark over the past years~\cite{hamilton2017representation}. 
%The recent emergence of the network embedding research has been largely triggered by the DeepWalk model~\cite{perozzi2014deepwalk}. 
In brief, DeepWalk is coupled with two steps. 
First, it generates several vertex sequences by random walks over a network; 
Second, it applies the skip-gram model~\cite{NIPS2013_5021} on the generated vertex sequences to learn the latent representations for each vertex. 
Commonly, skip-gram is parameterized with the context window size $T$ and the number of negative samples $b$.  
%$T$ and $b$ are the context window size and the number of negative sampling in skip-gram model, respectively. We finally introduce the definition of element-wise matrix function
Recently, a theoretical study~\cite{qiu2018network} reveals that DeepWalk essentially factorizes a matrix derived from the random walk process.
More formally, it proves that when the length of random walks goes to infinity, DeepWalk implicitly and asymptotically factorizes the following matrix:
%A recent study ~\cite{qiu2018network} shows the connection between DeepWalk and matrix factorization by proving the following theorem.
\beal{\label{eq:deepwalk}
\hlog\left(\frac{\vol(G)}{b}\bm{M}\right),
}where $\vol{(G)}=\sum_i \sum_j \bm{A}_{ij}$ denotes the volume of the graph, and \beq{
\label{eq:M}
\bm{M}= \frac{1}{T}\sum_{r=1}^T (\bm{D}^{-1}\bm{A})^r \bm{D}^{-1},
}where $\bm{D} = \diag{(d_1, \cdots, d_{n})}$ is the degree matrix with $d_i=\sum_{j}\bm{A}_{ij}$ as the generalized degree of the $i$-th vertex.  
Note that $\hlog(\bm{\cdot})$ represents the element-wise matrix logarithm~\cite{horn_johnson_1991}, which is different from the matrix logarithm. 
In other words, the matrix in Eq.~\eqref{eq:deepwalk} can be characterized as the result of applying element-wise matrix logarithm~(i.e., $\hlog$) to matrix $\frac{\vol{(G)}}{b}\bm{M}$.
%to a random-walk matrix-polynomial. 
% This conclusion leads to the understanding of skip-gram based network embedding models. 

The matrix in Eq.~(\ref{eq:deepwalk}) offers an alternative view of the skip-gram based network embedding methods. Further, Qiu et al. provide an explicit matrix factorization approach named NetMF to learn the embeddings~\cite{qiu2018network}. It shows that the accuracy for vertex classification based on the embeddings from NetMF outperforms that based on DeepWalk and LINE. Note that the matrix in Eq.~\eqref{eq:deepwalk} would be ill-defined if there exist a pair of vertices unreachable in $T$ hops, because $\log(0)=-\infty$. So following Levy and Goldberg~\cite{NIPS2014_5477}, NetMF uses the logarithm truncated at point one, that is, $\tln(x) = \max(0, \log(x))$. %, which is a sparse and consistent alternative. 
Thus, NetMF targets to factorize the matrix 
% \beq{
% \nonumber
% \htln\left( \frac{\vol(G)}{b}\left(\frac{1}{T}\sum_{r=1}^T (\bm{D}^{-1}\bm{A})^r\right) \bm{D}^{-1}\right).
% } 
\beal{
    \label{eq:deepwalk_matrix}
\htln\left(\frac{\vol(G)}{b}\bm{M}\right).
}In the rest of this work, we refer to the matrix in Eq. \eqref{eq:deepwalk_matrix} as \emph{the NetMF matrix}.

However, there exist a couple of challenges when leveraging the NetMF matrix in practice. First, almost every pair of vertices within distance $r\leq T$ correspond to a non-zero entry in the NetMF matrix.
% First, to precisely compute the matrix power $(\bm{D}^{-1}\bm{A})^r$, we have to enumerate all length-$r$ paths in the network, which is  extremely expensive in both time and space cost.  
% Second, 
Recall that many social and information networks exhibit the small-world property where most vertices can be reached from each other in a small number of steps. For example, as of the year 2012, 92\% of the reachable pairs in Facebook are at distance five or less~\cite{backstrom2012four}.
% ; In Open Academic Graph, more than 90\% of author pairs are at distance smaller than eight~(See \figref{fig:mag_distance} for details. %\yx{could remove the fig later if no space left}). 
As a consequence, even if setting a moderate context window size~(e.g., the default setting $T=10$ in DeepWalk), the NetMF matrix in Eq.~\eqref{eq:deepwalk_matrix} would be a dense matrix with $O(n^2)$ number of non-zeros. The exact construction and factorization of such a matrix is impractical for large-scale networks. More concretely, computing the matrix power in Eq.~\eqref{eq:M} involves dense matrix multiplication which costs $O(n^3)$ time; factorizing a $n\times n$ dense matrix is also time consuming. 
To reduce the construction cost, 
NetMF %attempts to 
approximates $\bm{M}$ with its top eigen pairs. %, which reduces the construction cost. 
However, the approximated matrix is still dense, making this strategy unable to handle large networks. 
%\jz{In NetMF~\cite{qiu2018network}, the authors propose to approximate $\bm{M}$ with its top eigen pairs to avoid  dense matrix produce which cost $O(n^3)$ time.  However, the approximator matrix is still a dense one, making NetMF difficult to scale to large-networks.} % and only for theoretical interest.

% Therefore, in this work, we aim to study large-scale network embedding with the balance between efficiency and effectiveness into consideration. The idea is to achieve a sparsification of the NetMF matrix with two requirements. 
% First, the sparsified matrix needs to be sparse enough, enabling its fast construction and factorization possible for very large-scale networks; 
% Second, it should be spectrally close to the dense NetMF matrix, guaranteeing that the spectral information of the network is maintained. 
% Thus, the sparsely learned embeddings could be as effective as those from the exact solution (NetMF), making it more outperform the empirical solutions (DeepWalk and LINE). 

In this work, we aim to address the efficiency and scalability limitation of NetMF, while maintaining its superiority in effectiveness. We list necessary notations and their descriptions in Table~\ref{tbl:notation}.

% \begin{figure}[t]
% \centering
% \includegraphics[width=.8\columnwidth]{figs/mag_distance.eps}
% \caption{The probability mass function of the distance distribution for Open Academic Graph. The average distance is 5.78. We use HyperANF~\cite{boldi2011hyperanf} to estimate this distribution.}
% \label{fig:mag_distance}
% \end{figure}

% \subsection{Problem Definition}
% In light of these challenges, we propose to study the fast solutions of the matrix factorization framework for large-scale network embedding. 
% In this work, we aim to search for a sparsification of the matrix in Eq. \ref{eq:deepwalk}.  
% More formally, \jz{add a relatively formal sentence}the objective is to find a sparsifier with significantly fewer non-zeros within xxxx time, error bound, etc. for the following matrix 
% \beq{\label{eq:deepwalk_matrix}
% \htln\left(\frac{\vol(G)}{b}\bm{M}\right)
% }with 
% \beq{
% \label{eq:M}
% \bm{M}= \frac{1}{T}\sum_{r=1}^T (\bm{D}^{-1}\bm{A})^r \bm{D}^{-1}.
% }. 
% In the rest of this work, we refer to the matrix in Eq. \ref{eq:deepwalk_matrix} as ``the NetMF matrix.''

\hide{
 For simplicity, we will denote 
\beq{
\label{eq:M}
\bm{M}= \frac{1}{T}\sum_{r=1}^T (\bm{D}^{-1}\bm{A})^r \bm{D}^{-1}.
}In the rest of this work. Thus the matrix we want to factorize finally becomes 
\beq{\label{eq:deepwalk_matrix}
\htln\left(\frac{\vol(G)}{b}\bm{M}\right).
}In the rest of this work, we will call the above matrix ``NetMF matrix''.

}%end of hide

%\input{preliminary.tex}
%\section{N\texorpdfstring{\MakeLowercase{et}}{et}SMF}
\section{Network Embedding as Sparse Matrix Factorization (N\texorpdfstring{\MakeLowercase{et}}{et}SMF)}
\label{sec:netsmf}

In this section, we develop network embedding as sparse matrix factorization~(NetSMF). 
We present the NetSMF method to construct and factorize a sparse matrix that approximates the dense NetMF matrix. 
The main technique we leverage is random-walk matrix-polynomial~(molynomial) sparsification.

%In this section, we develop NetSMF. As the exact computation of the NetMF matrix is intractable, we aim to find a sparse matrix which approximates it. The main technique we leverage is random-walk matrix-polynomial~(molynomial) sparsification.

\hide{
In the rest of this work, we assume that we use a shifted logarithm, $\ln1p(x) = \log(1+x)$, or a truncated logarithm, $\tln(x) = \max(0, \log(x))$, instead of the origin one   , to prevent the bad behavior of logarithm near point 0~($\log{0}=-\infty$). Thus, the matrix we want to factorize is 
\beq{
\nonumber
\besp{
\hln1p\left( \frac{\vol(G)}{b}\left(\frac{1}{T}\sum_{r=1}^T (\bm{D}^{-1}\bm{A})^r\right) \bm{D}^{-1}\right)\\
\htln\left( \frac{\vol(G)}{b}\left(\frac{1}{T}\sum_{r=1}^T (\bm{D}^{-1}\bm{A})^r\right) \bm{D}^{-1}\right)
}
} where $\hln1p$ denotes Hadamard $\ln1p$, a.k.a., element-wise $\ln1p$.
}

\subsection{Random-Walk Molynomial Sparsification}
\label{sec:sparsification}

We first introduce the definition of spectral similarity and the theorem of random-walk molynomial sparsification.

\begin{definition}{(Spectral Similarity of Networks)}
Suppose $G=(V, E, \bm{A})$  and $\widetilde{G}=(V, \widetilde{E}, \widetilde{\bm{A}})$ are two weighted undirected networks. Let  $\bm{L}=\bm{D}_{G} - \bm{A}$ and 
$\widetilde{\bm{L}}=\bm{D}_{\widetilde{G}} - \widetilde{\bm{A}}$ be their Laplacian matrices, respectively. We define $G$ and $\widetilde{G}$ are $(1+\epsilon)$-spectrally similar if 
\beq{\nonumber
\forall \bm{x}\in\mathbb{R}^{n}, (1-\epsilon)\cdot \bm{x}^\top \widetilde{\bm{L}} \bm{x} \leq \bm{x}^\top\bm{L} \bm{x} \leq (1+\epsilon) \cdot \bm{x}^\top \widetilde{\bm{L}} \bm{x}.
}
\end{definition}

\begin{theorem}{(Spectral Sparsifiers of Random-Walk Molynomials~\cite{cheng2015efficient,cheng2015spectral})}
\label{thm:sparsification}
For random-walk molynomial 
\beq{
\label{eq:L}
\bm{L} = \bm{D}- \sum_{r=1}^T \alpha_r \bm{D} \left(\bm{D}^{-1}\bm{A}\right)^r,
}where $\sum_{r=1}^T \alpha_r=1$ and $\alpha_r$ non-negative, one can construct, in time $O(T^2 m \epsilon^{-2} \log^2 n )$, a $(1+\epsilon)$-spectral sparsifier, $\widetilde{\bm{L}}$, with $O(n\log{n}\epsilon^{-2})$ non-zeros. For unweighted graphs, the time complexity can be reduced to $O(T^2 m\epsilon^{-2} \log n )$.
\end{theorem}

To achieve a sparsifier $\widetilde{\bm{L}}$ with $O(\epsilon^{-2}n\log{n})$ non-zeros, the sparsification algorithm consists of two steps:
\hide{
\begin{itemize}
\item The first step obtains an initial sparsifier for $\bm{L}$ with $O(T m \epsilon^{-2}\log{n}  )$ non-zeros. 
\item The second step
then applies the standard spectral sparsification algorithm~\cite{spielman2011graph} to further reduce the number of non-zeros to $O(\epsilon^{-2}n\log{n} )$.
\end{itemize}
}
The first step obtains an initial sparsifier for $\bm{L}$ with $O(T m \epsilon^{-2}\log{n}  )$ non-zeros. 
The second step
then applies the standard spectral sparsification algorithm~\cite{spielman2011graph} to further reduce the number of non-zeros to $O(\epsilon^{-2}n\log{n} )$. In this work, we only adopt the first step because a sparsifier with  $O(T m\epsilon^{-2} \log{n}  )$ non-zeros is sparse enough for our task. Thus we skip the second step that involves additional computations. From now on, when referring to the random-walk molynomial sparsification algorithm in this work, we mean its first step only.

\hide{
Secondly, according to our test,
the second step is somtimes not numerically stable enough. In our experiment, 
% I apply the sparsification algorithm implemented by PyGSP~\cite{michael_defferrard_2017_1003158}. However,  
it sometimes produces NaN~(not a number) during computation.
\end{itemize}

}

One can immediately observe that, if we set $\alpha_r=\frac{1}{T}, r\in [T]$, the matrix $\bm{L}$ in Eq.~\eqref{eq:L} has a strong connection with the desired matrix $\bm{M}$ in Eq.~\eqref{eq:M}. Formally, we have the following equation
\beq{
\label{eq:transformation}
\bm{M}=\bm{D}^{-1}\left(\bm{D}-\bm{L}\right)\bm{D}^{-1}.
}Thm.~\ref{thm:sparsification} can help us construct a sparsifier $\widetilde{\bm{L}}$ for matrix $\bm{L}$. Then we define
$\widetilde{\bm{M}}\triangleq \bm{D}^{-1}(\bm{D}-\widetilde{\bm{L}})\bm{D}^{-1}$ by 
replacing $\bm{L}$ in Eq.~\eqref{eq:transformation} with its sparsifier $\widetilde{\bm{L}}$. One can observe that matrix $\widetilde{\bm{M}}$ is still a sparse one with the same order of magnitude of non-zeros as $\widetilde{\bm{L}}$. Consequently, instead of factorizing the dense NetMF matrix in Eq.~\eqref{eq:deepwalk_matrix},  we can factorize its sparse alternative, i.e.,
\beq{
\label{eq:sparse_deepwalk}
\htln\left(\frac{\vol(G)}{b}\widetilde{\bm{M}}\right).
}In the rest of this work, the matrix in Eq.~\eqref{eq:sparse_deepwalk} is referred to as \emph{the NetMF matrix sparsifier}.

\subsection{The N\texorpdfstring{\MakeLowercase{et}}{et}SMF Algorithm}

In this section, we formally describe the NetSMF algorithm, which consists of three steps: 
random-walk molynomial sparsification, NetMF sparsifier construction, and truncated singular value decomposition.
%path sampling, sparsifier construction, and matrix factorization. %We now introduce them one by one, and discuss their time/space complexity later.

\begin{algorithm}[htp]
\small
  \SetAlgoLined\DontPrintSemicolon
  \SetKwInOut{Input}{Input}
  \SetKwInOut{Output}{Output}
  \SetKwFunction{ps}{PathSampling}
  \SetKwFunction{redsvd}{RandomizedSVD}
  \setcounter{AlgoLine}{0}
  \Input{ A social network $G=(V, E, \bm{A})$ which we want to learn network embedding; The number of non-zeros  $M$ in the  sparsifier; The dimension of embedding $d$.}
  \Output{ An embedding matrix of size $n\times d$, each row corresponding to a vertex.}
  $\widetilde{G} \gets (V, \emptyset, \widetilde{\bm{A}} = \bm{0})$\; \tcc{Create an empty network with $E=\emptyset$ and $\widetilde{\bm{A}}=0$.}
  \For{$i\gets1$ \KwTo $M$}{
    Uniformly pick an edge $e=(u,v) \in E$\;
    Uniformly pick an integer $r \in [T]$\;
    $u', v', Z \gets$ \ps{e, r}\;
    Add an edge $\left(u', v', \frac{2rm}{MZ}\right)$ to $\widetilde{G}$\;
    \tcc{Parallel edges will be merged into one edge, with their weights summed up together.}
  }
  Compute $\widetilde{\bm{L}}$ to be the unnormalized graph Laplacian of $\widetilde{G}$\;
  Compute $\widetilde{\bm{M}} = \bm{D}^{-1} \left(\bm{D} -  \widetilde{\bm{L}} \right)\bm{D}^{-1}$\;
  $\bm{U}_d, \bm{\Sigma}_d, \bm{V}_d \gets$ \redsvd{$\htln\left(\frac{\vol(G)}{b}\widetilde{\bm{M}}\right), d$}\; 
  %Rank-$d$ truncated SVD: $\htln\left(\frac{\vol(G)}{b}\tilde{\bm{M}}\right)\approx \bm{U}_d \bm{\Sigma}_d \bm{V}_d^\top $\;
  \Return $\bm{U}_d\sqrt{\bm{\Sigma}_d}$ as network embeddings\;
\normalsize
  \caption{NetSMF}
  \label{alg:netsmf}
\end{algorithm} 

\begin{algorithm}[htp]
\small
\SetAlgoLined\DontPrintSemicolon
\SetKwFunction{ps}{PathSampling}
  \SetKwProg{myalg}{Procedure}{}{}
  \myalg{\ps{$e=(u, v)$, $r$}}{
   Uniformly pick an integer $k \in [r]$\;
   Perform $(k-1)$-step random walk from $u$ to $u_0$\;
   Perform $(r-k)$-step random walk from $v$ to $u_r$\;
   Keep track of $Z(\bm{p})$ along the  length-$r$ path $\bm{p}$ between $u_0$ and  $u_r$ according to Eq.~\eqref{eq:Z}\;
   \KwRet $u_0, u_r, Z(\bm{p})$\;
  }{}
\normalsize
  \caption{\textsf{PathSampling} algorithm as described in \cite{cheng2015spectral}.}
  \label{alg:ps}
\end{algorithm}

\vpara{Step 1: Random-Walk Molynomial Sparsification.} To achieve the sparsifier $\widetilde{\bm{L}}$,  we adopt the algorithm in \citet{cheng2015spectral}.
The algorithm starts from creating a network $\widetilde{G}$ that has the same vertex set as $G$ and an empty edge set~(Alg.~\ref{alg:netsmf}, Line 1).
Next, the algorithm constructs a sparsifier with $O(M)$ non-zeros by repeating the \textsf{PathSampling} algorithm for $M$ times. In each iteration, it picks an edge $e\in E$ and an integer $r\in [T]$ uniformly~(Alg.~\ref{alg:netsmf}, Line 3-4). Then, the algorithm uniformly draws an integer $k\in [r]$ and performs $(k-1)$-step and $(r-k)$-step random walks starting from the two endpoints of edge $e$ respectively~(Alg.~\ref{alg:ps}, Line 3-4). The above process samples a length-$r$ path $\bm{p}=(u_0, u_1, \cdots, u_r)$. At the same time, the algorithm  keeps track of $Z(\bm{p})$, which is defined by
\beq{\label{eq:Z}
Z(\bm{p})= \sum_{i=1}^{r} \frac{2}{\bm{A}_{u_{i-1}, u_i}},
}and then adds a new edge $(u_0, u_r)$ with weight $\frac{2rm}{MZ(\bm{p})}$ to $\widetilde{G}$~(Alg.~\ref{alg:netsmf}, Line 6).\footnote{Details about how the edge weight is derived can be found in Thm.~\ref{thm:edge} in Appendix.} Parallel edges in $\widetilde{G}$ will be merged into one single edge, with
their weights summed up together.
Finally, the algorithm computes the Laplacian of $\widetilde{G}$, which is the sparsifier $\widetilde{\bm{L}}$ as we desired~(Alg.~\ref{alg:netsmf}, Line 8).
This step gives us a sparsifier with $O(M)$ non-zeros. %Strictly speaking, we need $M=O(Tm\log{n})$ to guarantee the approximation error. 

\vpara{Step 2: Construct a NetMF Matrix Sparsifier.} As we have discussed at the end of Section~\ref{sec:sparsification}, after constructing a sparsifier $\widetilde{\bm{L}}$, we can plug it into Eq.~\eqref{eq:transformation} to  obtain a  NetMF matrix sparsifier as shown in Eq.~\eqref{eq:sparse_deepwalk}~(Alg.~\ref{alg:netsmf}, Line 9-10).
This step does not change the order of magnitude of non-zeros in the sparsifier.

\begin{algorithm}[htp]
\small
  \SetAlgoLined \DontPrintSemicolon
  \SetKwInOut{Input}{Input}
  \SetKwInOut{Output}{Output}
  \SetKwFunction{redsvd}{RandomizedSVD}
  \SetKwProg{myalg}{Procedure}{}{}
  \tcc{In this work, the matrix to be factorized~(Eq.~\eqref{eq:sparse_deepwalk}) is an $n\times n$ symmetric sparse matrix. We store this sparse matrix in a row-major way and make use of its symmetry to simplify the computation.}
  \myalg{\redsvd{$\bm{A}$, $d$}}{
   Sampling Gaussian random matrix $\bm{O}$ \tcp*{$\bm{O} \in \mathbb{R}^{n\times d}$}
   Compute sample matrix $\bm{Y} = \bm{A}^\top \bm{O} = \bm{A} \bm{O}$ \tcp*{$\bm{Y} \in \mathbb{R}^{n\times d}$}
   Orthonormalize $\bm{Y}$\;
   Compute $\bm{B} = \bm{A} \bm{Y}$ \tcp*{$\bm{B} \in \mathbb{R}^{n\times d}$}
   Sample another Gaussian random matrix $\bm{P}$  \tcp*{$\bm{P} \in \mathbb{R}^{d\times d}$}
   Compute sample matrix of $\bm{Z} = \bm{B}\bm{P}$  \tcp*{$\bm{Z} \in \mathbb{R}^{n\times d}$}
   Orthonormalize $\bm{Z}$\;
   Compute $\bm{C} = \bm{Z}^\top\bm{B}$ \tcp*{$\bm{C} \in \mathbb{R}^{d\times d}$}
   Run Jacobi SVD on $\bm{C}=\bm{U}\bm{\Sigma}\bm{V}^\top$\;
   \KwRet $\bm{Z}\bm{U}$, $\bm{\Sigma}$, $\bm{Y}\bm{V}$\; 
  }{}
   \setcounter{AlgoLine}{0}
  \tcc{Result matrices are of shape $n\times d, d\times d, n \times d$ resp.}
\normalsize
  \caption{Randomized SVD on NetMF Matrix Sparsifier}
    \label{alg:redsvd}
\end{algorithm}

\begin{table}[t!]
\centering
\caption{Time and Space Complexity of NetSMF.}
\label{tbl:complexity}
\centering \small
\begin{tabular}[htbp]{c|c|c}
\hline \hline
& \textbf{Time} & \textbf{Space}  \\ \hline
\textbf{Step 1} & 
\begin{tabular}{c}$O(MT\log{n})$ for weighted networks \\$O(MT)$ for unweighted networks\end{tabular}
 & $O(M+n+m)$\\ \hline
\textbf{Step 2} & $O(M)$ & $O(M+n)$ \\ \hline
\textbf{Step 3} & $O(Md+nd^2+d^3)$ & $O(M+nd)$ \\
\hline\hline
\end{tabular}
\normalsize
\end{table}

\vpara{Step 3: Truncated Singular Value Decomposition.} The final step is to perform truncated singular value decomposition~(SVD) on the constructed NetMF matrix sparsifier~(Eq.~\eqref{eq:sparse_deepwalk}).  However, even the sparsifier only has $O(M)$ number of non-zeros, performing exact SVD is still time consuming.
In this work, we leverage a modern randomized matrix approximation technique---Randomized SVD---developed by \citet{halko2011finding}.
% We don't plan to go into too much details about this technique in this paper. 
Due to space constraint, we cannot include many details. Briefly speaking, the algorithm projects the original matrix to a low-dimensional space through a Gaussian random matrix. One only needs to perform traditional SVD~(e.g. Jacobi SVD) on a $d\times d$ small matrix. 
We list the pseudocode algorithm in Alg.~\ref{alg:redsvd}.
Another advantage of SVD is that we can determine the dimensionality of embeddings by using, for example, Cattell's Scree test~\cite{cattell1966scree}. In the test, we plot the singular values and select a rank $d$ such that there is a clear drop in the magnitudes or 
the singular values start to even out. More details will be discussed in Section~\ref{sec:exp}.

\vpara{Complexity Analysis.} Now we analyze the time and space complexity of NetSMF, as summarized in Table~\ref{tbl:complexity}.
As for step 1,
we call the $\textsf{PathSampling}$ algorithm for $M$ times, during each of which it performs $O(T)$ steps of random walks over the network. For unweighted networks, sampling a neighbor requires $O(1)$ time, while for weighted networks, one can use roulette wheel selection to choose a neighbor in $O(\log{n})$. It taks $O(M)$ space to store $\widetilde{G}$, while the additional $O(n+m)$ space comes from the storage of the input network.
As for step 2, it takes $O(M)$ time to perform the transformation in Eq.~\eqref{eq:transformation} and the element-wise truncated logarithm in Eq.~\eqref{eq:sparse_deepwalk}. The additional $O(n)$ space is spent in storing the degree matrix.
As for step 3, $O(Md)$ time is required to compute the product of  a row-major sparse matrix and  a dense matrix~(Alg.~\ref{alg:redsvd}, Lines 3 and 5); $O(nd^2)$ time is spent in Gram-Schmidt orthogonalization~(Alg.~\ref{alg:redsvd}, Lines 4 and 8); $O(d^3)$ time is spent in Jacobi SVD~(Alg.~\ref{alg:redsvd}, Line 10).

\begin{figure*}[ht!]
	\centering
 \includegraphics[width=1.\textwidth]{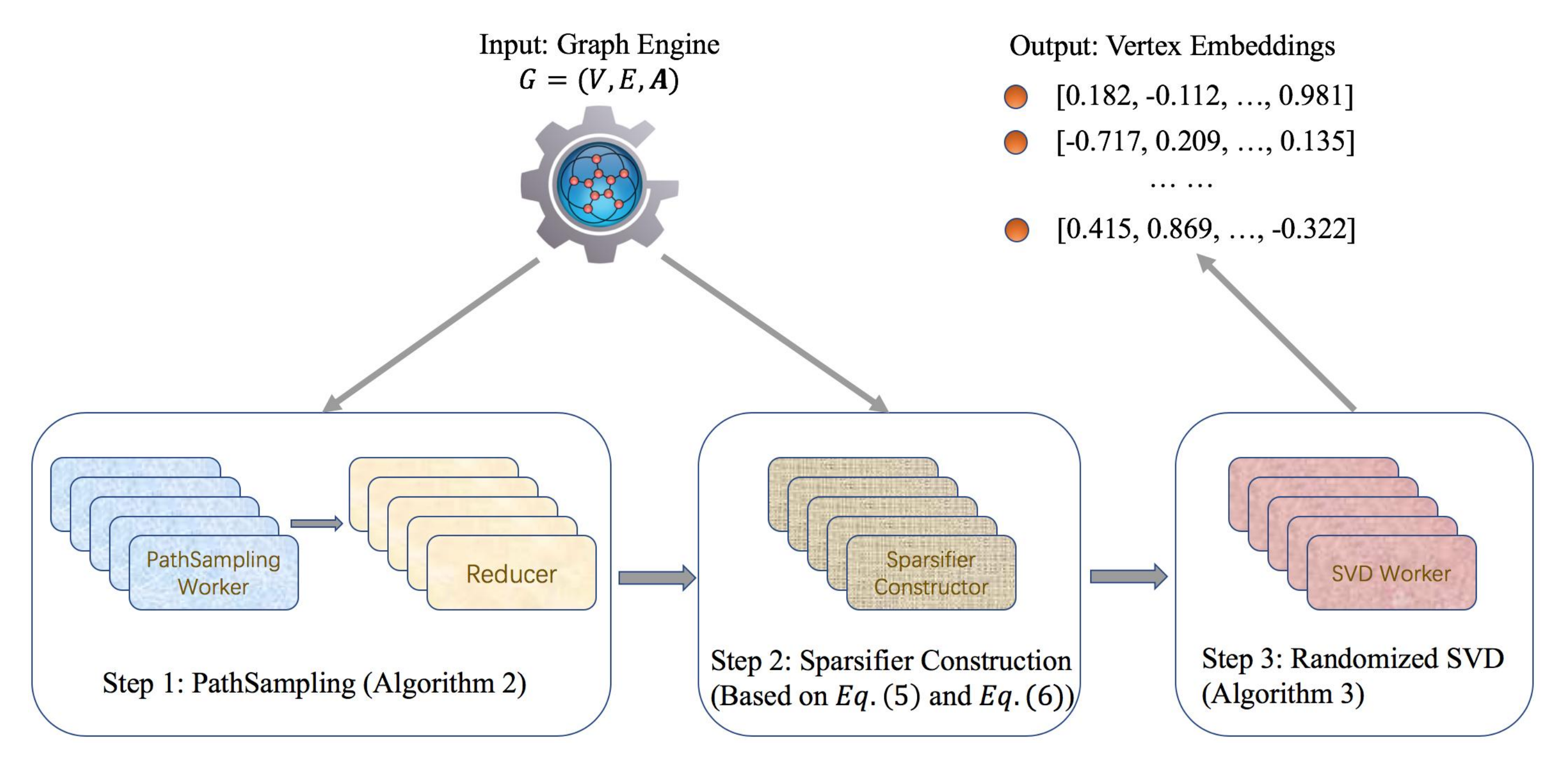}
	\caption{The System Design of NetSMF.  \textmd{The input comes from a graph engine which stores the network data and provides efficient APIs to graph queries.
	In Step 1, the system launches several \textsf{PathSampling} workers. Each worker handles a subset of samples. Then, a reducer is designed to aggregate the output of the \textsf{PathSampling} algorithm. In Step 2, the system distributes data to several sparsifier constructors to perform the transformation defined in Eq.~\eqref{eq:transformation} and the truncated element-wise matrix logarithm in Eq.~\eqref{eq:sparse_deepwalk}.
	In the final step, the system applies truncated randomized SVD on the constructed sparsifier and dumps the resulted embeddings to storage.
	}}
	\label{fig:sys}
\end{figure*}

\vpara{Connection to NetMF.}
The major difference between NetMF and NetSMF lies in the approximation strategy of the NetMF matrix in Eq.~\eqref{eq:deepwalk_matrix}. 
As we mentioned in Section~\ref{sec:pre}, NetMF approximates it with a dense matrix, which brings new space and computation challenges. 
In this work, NetSMF aims to find a sparse approximator to the NetMF matrix by leveraging theories and techniques from spectral graph sparsification.

\vpara{Example.} We provide a running example to help  understand the NetSMF algorithm.
Suppose we want to learn embeddings for a network with $n=10^6$ vertices, $m=10^7$ edges, context window size $T=10$, and approximation factor $\epsilon=0.1$. The NetSMF method calls the \textsf{PathSampling} algorithm for  $M=Tm\epsilon^{-2}\log{n}\approx 1.4\times 10^{11}$  times and provides us with a NetMF matrix sparsifier with at most $1.4\times 10^{11}$ non-zeros~(Notice that the reducer in Step 1 and $\htln$ in Step 2 will further sparsify the matrix, making $1.4\times 10^{11}$ an upper bound). The density of the sparsifier is at most $\frac{M}{n^2} \approx 14\%$. Then, when computing the sparse-dense matrix product in randomized SVD~(Alg.~\ref{alg:redsvd}, Lines 3 and 5), the sparseness of the factorized matrix can greatly accelerate the calculation. In comparison, NetMF must construct a dense matrix with $n^2=10^{12}$ non-zeros, which is an order of magnitude larger in terms of density. Also, the density of the sparsifier in NetSMF can be further reduced by using a larger $\epsilon$, while NetMF does not have this flexibility.

\subsection{Approximation Error Analysis}
\label{sec:error}
In this section, we analyze the approximation error of the sparsification. 
We assume that we choose an approximation factor $\epsilon < 0.5$.
We first see how the constructed $\widetilde{\bm{M}}$
approximates $\bm{M}$ and then compare the NetMF matrix~(Eq.~\eqref{eq:deepwalk_matrix}) against the NetMF matrix sparsifier~(Eq.~\eqref{eq:sparse_deepwalk}). 
We use $\sigma_i$ to denote the $i$-th descending-order singular value of a matrix. We also assume the  vertices' degrees are sorted in ascending order, that is, $d_{\min}=d_1 \leq d_2 \cdots \leq d_n$. %Specially, we denote the minimum vertex degree to be $d_{\min}$~(i.e., $d_{\min}=d_1$).

\begin{restatable}{theorem}{Merror}
\label{thm:Merror}
  The  singular value of $\widetilde{\bm{M}}-\bm{M}$ satisfies $
\sigma_i(\widetilde{\bm{M}}-\bm{M}) \leq \frac{4\epsilon}{\sqrt{d_i d_{\min}}}, \forall i \in [n]
$.
\end{restatable}
%\begin{proof}
%See Appendix.
%\end{proof}

\begin{restatable}{theorem}{logerror}
\label{thm:log_error}
Let $\norm{\cdot}{F}$ be the matrix Frobenius norm. Then
\beq{
\nonumber
\norm{\htln\left(\frac{\vol(G)}{b}\widetilde{\bm{M}}\right)-
\htln\left(\frac{\vol(G)}{b}\bm{M}\right)}{F}
 \leq \frac{4\epsilon\vol(G)}{b\sqrt{d_{\min}}} \sqrt{\sum_{i=1}^n \frac{1}{d_i}}.
}
\end{restatable}
\begin{proof}
See Appendix.
\end{proof}

\vpara{Discussion on the Approximation Error.}
%The above bound might be loose, since no assumption about the network is made at all. 
The above bound is achieved without making assumptions about the input network. 
If we introduce some assumptions, say a bounded lowest degree $d_{\min}$  or a specific random graph model~(e.g.,  Planted Partition Model or  Extended Planted Partition Model), it is promising to explore tighter bounds by leveraging theorems 
 in literature \cite{dasgupta2004spectral, chaudhuri2012spectral}.

%
% Here are some links which may be useful\footnote{\url{http://www.cs.yale.edu/homes/spielman/561/2009/lect21-09.pdf}}\footnote{\url{http://www.cs.columbia.edu/~djhsu/coms4772-f16/lectures/planted.md.handout.pdf}}.

% briefly introduce randomized svd algorithm~\cite{halko2011finding}.

\subsection{Parallelization}
\label{sec:sys}

Each step of NetSMF can be  parallelized, enabling it to scale to very large networks. The parallelization  design of NetSMF is introduced in Figure~\ref{fig:sys}. Below we discuss the parallelization of each step in detail.  
At the first step, the paths in the \textsf{PathSampling} algorithm are sampled independently with each other. 
 Thus we can launch multiple \textsf{PathSampling} workers simultaneously. Each worker handles a subset of the samples. 
Herein, we require that each worker is able to access the network data $G=(V, E, \bm{A})$ efficiently. 
There are many options to meet this requirement. 
The easiest one is to load a copy of the network data to each worker's memory.  When the network is extremely large~(e.g., trillion scale) or  workers have memory constraints, the graph engine should be designed to expose efficient graph query APIs %~\cite{shao2013trinity} 
to support graph operations such as random walks. 
At the end of this step, a reducer is designed to merge parallel edges and sum up their weights. If this step is implemented in a big data system such as Spark~\cite{zaharia2010spark},  % and SCOPE~\cite{chaiken2008scope}, 
the reduction step can be simply achieved by running a \textsf{reduceByKey(\_+\_)}\footnote{\url{https://spark.apache.org/docs/latest/rdd-programming-guide.html}} function.
After the reduction, the sparsifier $\widetilde{\bm{L}}$ is organized as a collection of triplets, a.k.a, COOrdinate format, with each indicating an entry of the sparsifier. 
The second step is the most straightforward step to scale up. When processing a triplet $(u, v, w)$,  we can simply query the degree of vertices $u$ and $v$ and perform the transformation defined in Eq.~\eqref{eq:transformation} as well as the truncated logarithm in Eq.~\eqref{eq:sparse_deepwalk}, which can be well parallelized. %
For the last step, we organize the sparsifier into row-major format. This format allows efficient multiplication between a sparse and a dense matrix~(Alg.~\ref{alg:redsvd}, Line 3 and 5). Other dense matrix operators~(e.g., Gaussian random matrix generation, Gram-Schmidt
orthogonalization and Jacobi SVD) can be easily accelerated by using multi-threading or common linear algebra libraries. % like BLAS, LAPACK and Intel MKL. 
%
%\vpara{Implementation Note} 
In this work, we adopt a single-machine shared-memory implementation. 
We use OpenMP~\cite{dagum1998openmp} to parallelize NetSMF in our implementation\footnote{Code is publicly available at \url{https://github.com/xptree/NetSMF}}. 
%The code will be made publicly available when the double-blind review period ends. 

%\footnote{OpenMP is a high-level and portable API for multi-platform shared-memory parallel programming}. 
%Our  implementation of randomized SVD is based on an open-source project\footnote{\url{https://code.google.com/archive/p/redsvd/}}.

%\input{system.tex}

\section{Experiments}
\label{sec:exp}
In this section, we evaluate the proposed NetSMF method on the  multi-label vertex classification task, which has been commonly used to evaluate previous network embedding techniques~\cite{perozzi2014deepwalk, tang2015line, grover2016node2vec, qiu2018network}. We introduce our datasets and baselines in Section~\ref{sec:data} and Section~\ref{sec:baseline}. We report experimental results and parameter analysis in Section~\ref{sec:result} and Section~\ref{sec:para}, respectively.

\begin{figure*}[!t]
	\centering
 \includegraphics[width=1.\textwidth]{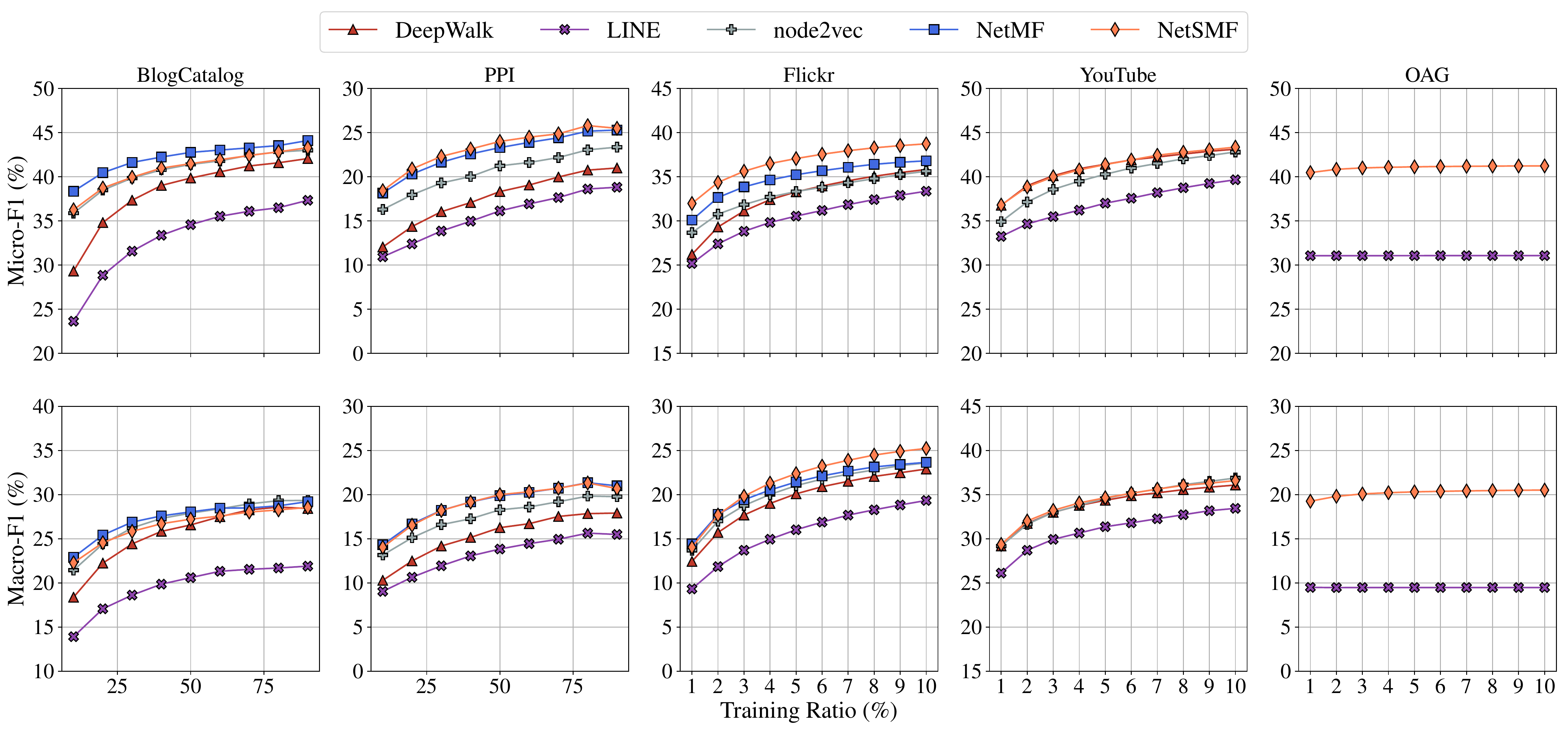}
	\caption{
Predictive	performance 
%on varying 
w.r.t.
the ratio of training data.  \textmd{The $x$-axis represents the
ratio of labeled data~(\%), and the $y$-axis in the top and bottom rows denote the Micro-F1 and Macro-F1 scores~(\%) respectively.  For methods which fail to finish computation in one week or cannot handle the computation, their results are not available and thus not plotted in this figure.}
	\normalsize} 
	\label{fig:f1}
\end{figure*}

\subsection{Datasets} 
\label{sec:data}

\begin{table}[t!]
\centering
\caption{Statistics of Datasets.}
\label{tbl:data}
\centering \small
\begin{tabular}[htbp]{c|c|c|c|c|c}
\hline \hline
\textbf{Dataset} & \textbf{BlogCatalog} & \textbf{PPI}   & \textbf{Flickr}  & \textbf{YouTube} & \textbf{OAG}\\ \hline
$\abs{V}$ & 10,312 & 3,890   &  80,513 & 1,138,499 & 67,768,244\\\hline
$\abs{E}$ & 333,983 & 76,584  &5,899,882 & 2,990,443 & 895,368,962\\\hline
\#labels & 39 & 50 & 195 & 47 & 19\\
\hline\hline
\end{tabular}
\normalsize
\end{table}	

We employ five datasets for the prediction task, four of which are in relatively small scale but have been widely used in network embedding literature, including BlogCatalog, PPI, Flickr, and YouTube. 
The remaining one is a large-scale academic co-authorship network, which is at least two orders of magnitude larger than the largest one (YouTube) used in most network embedding studies.  
%Besides them, we also propose to use a large-scale co-authorship network
%which is two orders of magnitude larger than the previous largest dataset.
The statistics of these datasets are listed in Table~\ref{tbl:data}. 

\vpara{BlogCatalog~\cite{tang2009relational, agarwal2009social}} is a network of social relationships of online bloggers. 
The vertex labels represent the interests of the bloggers. %~(categories provided by the authors).

\vpara{Protein-Protein Interactions~(PPI)~\cite{stark2010biogrid}} is a subgraph of the 
 PPI network for Homo Sapiens. The vertex labels are obtained from the hallmark gene sets and represent
biological states. 

\hide{
{\bf  Wikipedia\footnote{\url{http://mattmahoney.net/dc/text.html}}} is a co-occurrence network of words
appearing in the first million bytes of the Wikipedia dump.
The labels are the Part-of-Speech~(POS) tags inferred
by Stanford POS-Tagger~\cite{toutanova2003feature}. 
}

\vpara{Flickr~\cite{tang2009relational}} is the user contact network 
in Flickr. The labels represent the
interest groups of the users.

\vpara{YouTube~\cite{tang2009scalable}} is a video-sharing website that allows users to upload, view, rate, share, add to their favorites, report, comment on videos. The users are labeled by the video genres they liked.
% The labels in this dataset are defined to be a collection of users like common video genres.

\vpara{Open Academic Graph (OAG)\footnote{\url{www.openacademic.ai/oag/}}} is an academic graph indexed by Microsoft Academic~\cite{sinha2015overview} and AMiner.org~\cite{tang2008arnetminer}. 
We construct an undirected co-authorship network from OAG, which contains 67,768,244 authors and 895,368,962 collaboration edges. 
The vertex labels are defined to be the top-level fields of study of each author, such as computer science, physics and psychology.
In total, there are 19 distinct fields~(labels) and authors may publish in more than one field, making the associated vertices have multiple labels. % 
%Authors may be associated with multiple fields
%In total, there are 19 distinct labels, including  
%, which are encoded in a a hierarchy such as ``cs.machine\_learning'' and ``math.functional\_analysis''. The first and the second level of the hierarchy contain 19 and 294 distinct classes, respectively. 
% In this work, we use the level-1 labels of the hierarchy.
% We will soon release this dataset. We hope this dataset acts as a benchmark for large-scale network embedding learning and semi-supervised learning on graphs~\cite{gao2018large, ying2018graph}.

\subsection{Baseline Methods} 
\label{sec:baseline}

We compare NetSMF with NetMF~\cite{qiu2018network}, 
LINE~\cite{tang2015line},
DeepWalk~\cite{perozzi2014deepwalk}, and node2vec~\cite{grover2016node2vec}.
For NetSMF, NetMF, DeepWalk, and node2vec that allow multi-hop structural dependencies, the context window size $T$ is set to be 10, which is also the default setting used in both DeepWalk and node2vec. 
Across all datasets, we set the embedding dimension $d$ to be 128. 
%For LINE, DeepWalk, and NetSMF that support multi-threading/processing acceleration, we set the number of threads/processes to be 30. 
We follow the common practice for the other hyper-parameter settings, which are introduced below. 
%In NetSMF, we set the number of samples $M=10^3\times T\times m$ for the PPI, Flickr, and YouTube datasets; 
%To achieve desired prediction performance, we set $M=10^4\times T \times m$ for BlogCatalog and $M=10\times T\times m$ for OAG; 
%For OAG, %limited by the memory of a single machine, 
%we set $M=10\times T\times m$. 
%For LINE, DeepWalk, and NetSMF that support multi-thread/process acceleration, we set the number of threads/processes to be 30. 
%Moreover, for NetMF and NetSMF, we set $b=1$.

\vpara{LINE.} We use LINE with the second order proximity~(i.e., LINE~(2nd) \cite{tang2015line}). We use the default setting of LINE's hyper-parameters: the number of edge samples to be 10 billion and the negative sample size to be 5. 

\vpara{DeepWalk.}
We present DeepWalk's results with the authors' preferred parameters, that is, walk length to be 40, the number of walks from each vertex  to be 80, and the number of negative samples in skip-gram to be 5. 

% We compromise to choose walk length to be 20 and the number of walks to be 5.

\vpara{node2vec.}  %for BlogCatalog and PPI that were used in the original 
For the return parameter $p$ and in-out parameter $q$ in node2vec, we adopt the default setting that was used by its authors if available. Otherwise, we grid search $p, q\in \{0.25, 0.5, 1, 2, 4\}$. For a fair comparison,  we use the same walk length and the number of walks per vertex as DeepWalk.

\vpara{NetMF.} 
In NetMF, the hyper-parameter $h$ indicates the number of eigen pairs used to approximate the NetMF matrix. 
%Following the authors' preferred setting~\cite{qiu2018network},
We choose  $h=256$ for the BlogCatalog, PPI and Flickr datasets.

\vpara{NetSMF.} 
In NetSMF, we set the number of samples $M=10^3\times T\times m$ for the PPI, Flickr, and YouTube datasets, 
 $M=10^4\times T \times m$ for BlogCatalog,  and $M=10\times T\times m$ for OAG in order to achieve desired performance. 
For both NetMF and NetSMF, we have $b=1$.

\vpara{Prediction Setting.} 
We follow the same experiment and evaluation procedures that were performed in DeepWalk~\cite{perozzi2014deepwalk}. First, we randomly sample a portion of labeled vertices for training and use the remaining for testing. For the BlogCatalog and PPI datasets, the training ratio varies from 10\% to 90\%. For Flickr, YouTube and OAG, the training ratio varies from 1\% to 10\%. 
%The same setting is also applied to the OAG dataset.  
%
We use the one-vs-rest logistic regression model implemented by LIBLINEAR~\cite{fan2008liblinear} for the multi-label vertex classification task. In the test phase, the one-vs-rest model yields a ranking of labels rather than an exact label assignment. To avoid the thresholding effect, %~\cite{tang2009large}, 
we take the assumption that was made in DeepWalk, LINE, and node2vec, that is, the number of labels for vertices in the test data is given~\cite{perozzi2014deepwalk,tang2009large,grover2016node2vec}. 
We repeat the prediction procedure ten times and evaluate the average performance 
%of different
%approaches 
in terms of both Micro-F1 and Macro-F1 scores~\cite{tsoumakas2009mining}.
All the experiments are performed on a server with Intel Xeon E7-8890 CPU (64 cores), 1.7TB memory, and 2TB SSD hard drive.

\subsection{Experimental Results}
\label{sec:result}

We summarize the prediction performance in \figref{fig:f1}. %and list the quantitative and relative gaps between our NetSMF and baselines in Table~\ref{tbl:f1}.
 To compare the efficiency of different algorithms, we also list the running time of each algorithm across all datasets, if available, in Table~\ref{tbl:efficiency}.

\begin{table}[t]
\caption{Efficiency comparison. \textmd{ The running time includes filesystem IO and computation time. ``--'' indicates that the corresponding algorithm fails to complete within one week. ``$\times$'' indicates that the corresponding algorithm is unable to handle the computation due to excessive space and memory consumption.}
}
\centering \small
\label{tbl:efficiency}
\begin{tabular}{c|c|c|c|c|c}
\hline \hline
      & \rot{\textbf{LINE}} & \rot{\textbf{DeepWalk}}  & \rot{\textbf{node2vec}} & \rot{\textbf{NetMF}} & \rot{\textbf{NetSMF}}\\ \hline
 BlogCatalog & 40 mins & 12 mins & 56 mins & 2 mins & 13 mins \\
 PPI & 41 mins & 4 mins & 4 mins & 16 secs  & 10 secs\\
 Flickr & 42 mins & 2.2 hours & 21 hours & 2 hours & 48 mins \\
 YouTube & 46 mins & 1 day & 4 days & $\times$ & 4.1 hours \\
 OAG & 2.6 hours & -- & -- & $\times$ & 24 hours \\  \hline
 \hline
\end{tabular}
\normalsize
\end{table}

\vpara{NetSMF vs. NetMF.} 
We first focus on the comparison between NetSMF and NetMF, since the goal of NetSMF is to address the efficiency and scalability issues of NetMF while maintaining its superiority in effectiveness. 
%aim to factorize (an approximation of) the NetMF matrix in Eq.~\eqref{eq:deepwalk_matrix}. 
From Table~\ref{tbl:efficiency}, we observe that for YouTube and OAG, both of which contain more than one million vertices, NetMF fails to complete because of the excessive space and memory consumption, while NetSMF is able to finish in four hours and one day, respectively. 
For the moderate-size network Flickr, both methods are able to complete within one week, though NetSMF is 2.5$\times$ faster~(i.e., 48 mins vs. 2 hours). 
For small-scale networks, NetMF is faster than NetSMF in BlogCatalog and is comparable to NetSMF in PPI in terms of running time. This is because when the input networks contain only thousands of vertices, the advantage of sparse matrix construction and factorization over its dense alternative could be marginalized by other components of the workflow.

In terms of prediction performance, \figref{fig:f1} suggests 
NetSMF and NetMF yield consistently the best results among all compared methods, empirically demonstrating the power of the matrix factorization framework for network embedding.
In BlogCatalog, NetSMF has slightly worse performance than NetMF~(on average less than 3.1\% worse regarding both Micro- and Macro-F1). 
In PPI, the two leading methods' performance are relatively indistinguishable in terms of both metrics. 
%, and the Micro-F1 of NetSMF is slightly better than NetMF with on average 2.3\% improvements.  
In Flickr, NetSMF achieves significantly better Macro-F1 than NetMF~(by 3.6\% on average), and also higher Micro-F1~(by 5.3\% on average). 
Recall that NetMF uses a dense approximation of the matrix to factorize. These results show that the sparse spectral approximation used by NetSMF does not necessarily yield worse performance than the dense approximation used by NetMF. 
%NetMF cannot handle the computation for the YouTube and OAG datasets. 

%\yx{Recall that NetMF cannot handle the computation of YouTube and OAG because its approximation matrix  cannot be stored in memory~(e.g. storing a $n\times n$ dense floating pint matrix costs >4TB for YouTube and \jz{xx} for OAG). }

\textit{
Overall, not only NetSMF improves the scalability, and the running time of NetMF by orders of magnitude for large-scale networks, it also has competitive, and sometimes better, performance. 
This demonstrates the effectiveness of our spectral sparsification based approximation algorithm. 
}

\vpara{NetSMF vs. DeepWalk, LINE $\&$ node2vec.}
We also compare NetSMF against common graph embedding benchmarks---DeepWalk, LINE, and node2vec. 
For the OAG dataset, DeepWalk and node2vec fail to finish the computation within one week, while NetSMF requires only 24 hours.  
Based on the publicly reported running time of skip-gram~\cite{mikolov2013efficient}, we estimate that DeepWalk and node2vec may require months to generate embeddings for the OAG dataset. 
In BlogCatalog, DeepWalk and NetSMF require similar computing time, while in Flickr, YouTube, and PPI, NetSMF is 2.75$\times$, 5.9$\times$, and 24$\times$ faster than DeepWalk, respectively. 
In all the datasets, NetSMF achieves 4--24$\times$ speedup over node2vec. 
%In all the datasets except BlogCatalog where NetSMF and DeepWalk require similar computing time,  NetSMF is faster than DeepWalk by 5-95\%. Compared to node2vec, NetSMF is faster by 4-43\%. 

Moreover, the performance of NetSMF is significantly better than DeepWalk in BlogCatalog, PPI, and Flickr, by 7--34\% in terms of Micro-F1 and 5--25\% in terms of Macro-F1. 
In YouTube, NetSMF achieves comparable results to DeepWalk. %~(by 0.1\% and 1\% in Micro and Macro F1, resp.).
Compared with node2vec, NetSMF achieves comparable performance in BlogCatalog and YouTube, and significantly better performance in PPI and Flickr. 
In summary, NetSMF consistently outperforms DeepWalk and node2vec in terms of both efficiency and effectiveness. 

LINE has the best efficiency among all the five methods and together with NetSMF, they are the only methods that can generate embeddings for OAG within one week (and both finish in one day). 
However, it also has the worst prediction performance and consistently loses to others by a large margin across all datasets. 
For example, NetSMF beats LINE by 21\% and 39\% in Flickr, and by 30\% and 100\% in OAG in terms of Micro-F1 and Macro-F1, respectively.  %, and DeepWalk outperforms LINE by 8\% and 25\%, in terms of Micro and Macro F1. 

In summary, LINE achieves efficiency at the cost of ignoring  multi-hop  dependencies in networks, which are supported by all the other four methods---DeepWalk, node2vec, NetMF, and NetSMF, demonstrating the importance of multi-hop dependencies for learning network representations. 

\textit{
More importantly, among these four methods, DeepWalk achieves neither efficiency nor effectiveness superiority; 
node2vec achieves relatively good performance at the cost of efficiency; 
NetMF achieves effectiveness at the expense of significantly increased time and space costs; 
 NetSMF is the only method that achieves both high efficiency and effectiveness, empowering it to learn effective embeddings for billion-scale networks (e.g., the OAG network with 0.9 billion edges) in one day on one modern server. 
}

\subsection{Parameter Analysis}
\label{sec:para}
\begin{figure*}[!t]
	\centering
	\mbox{
	\subfigure[]{
	    \includegraphics[width=.49\columnwidth]{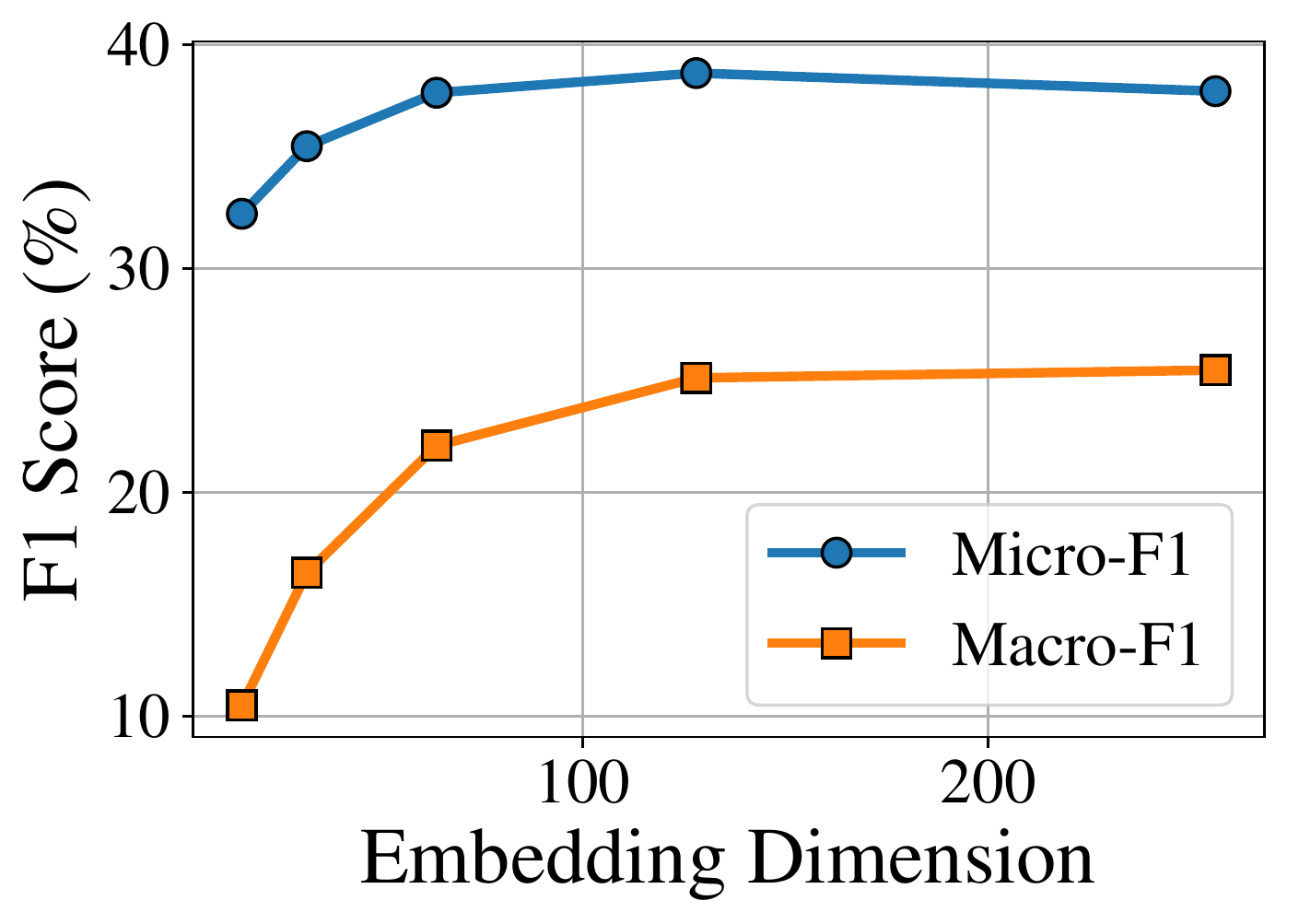}
	    \label{fig:dim}
	}
	\subfigure[]{
	    \includegraphics[width=.53\columnwidth]{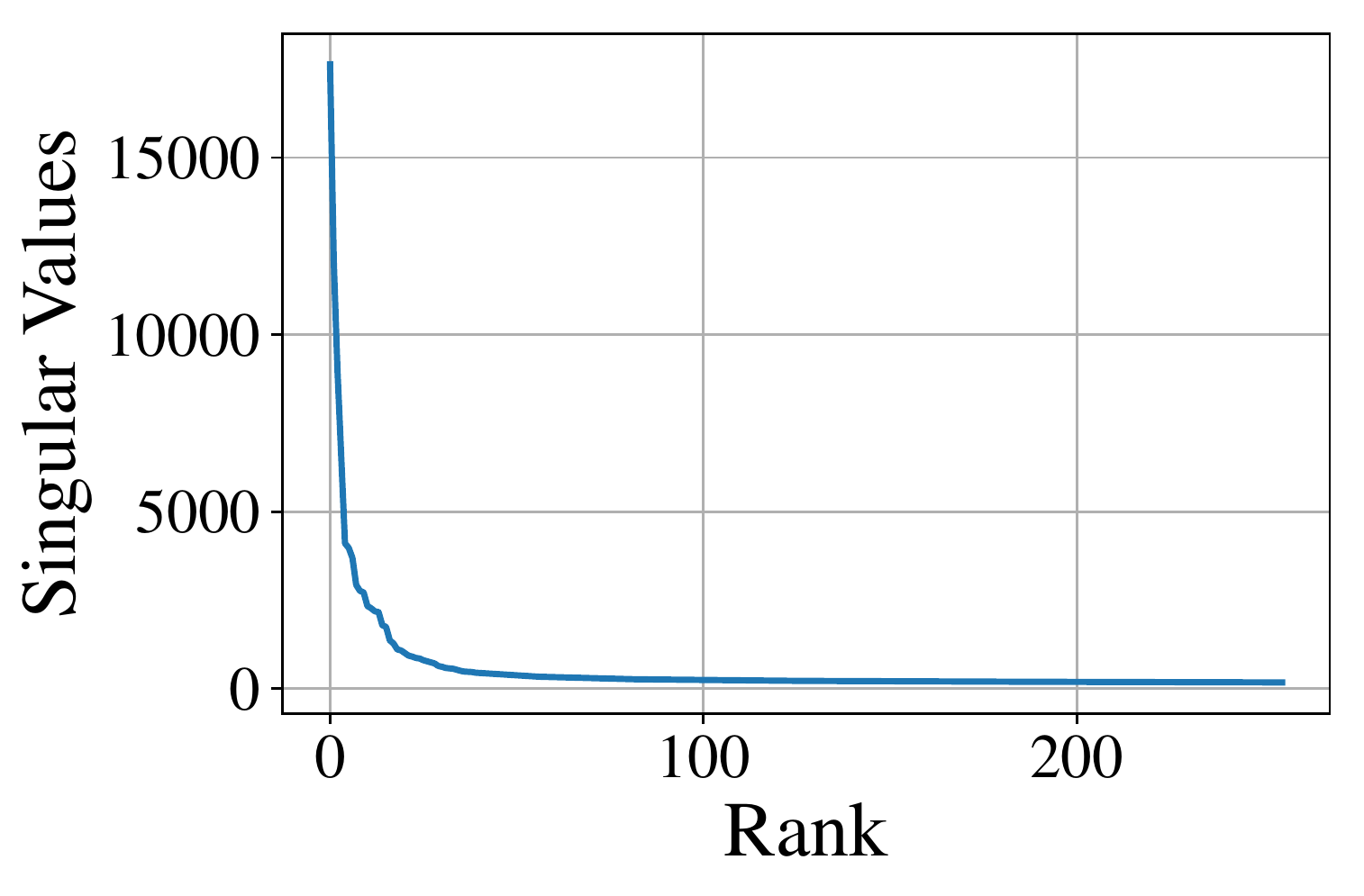}
	    \label{fig:singular}
	}
	\subfigure[]{
	    \includegraphics[width=.49\columnwidth]{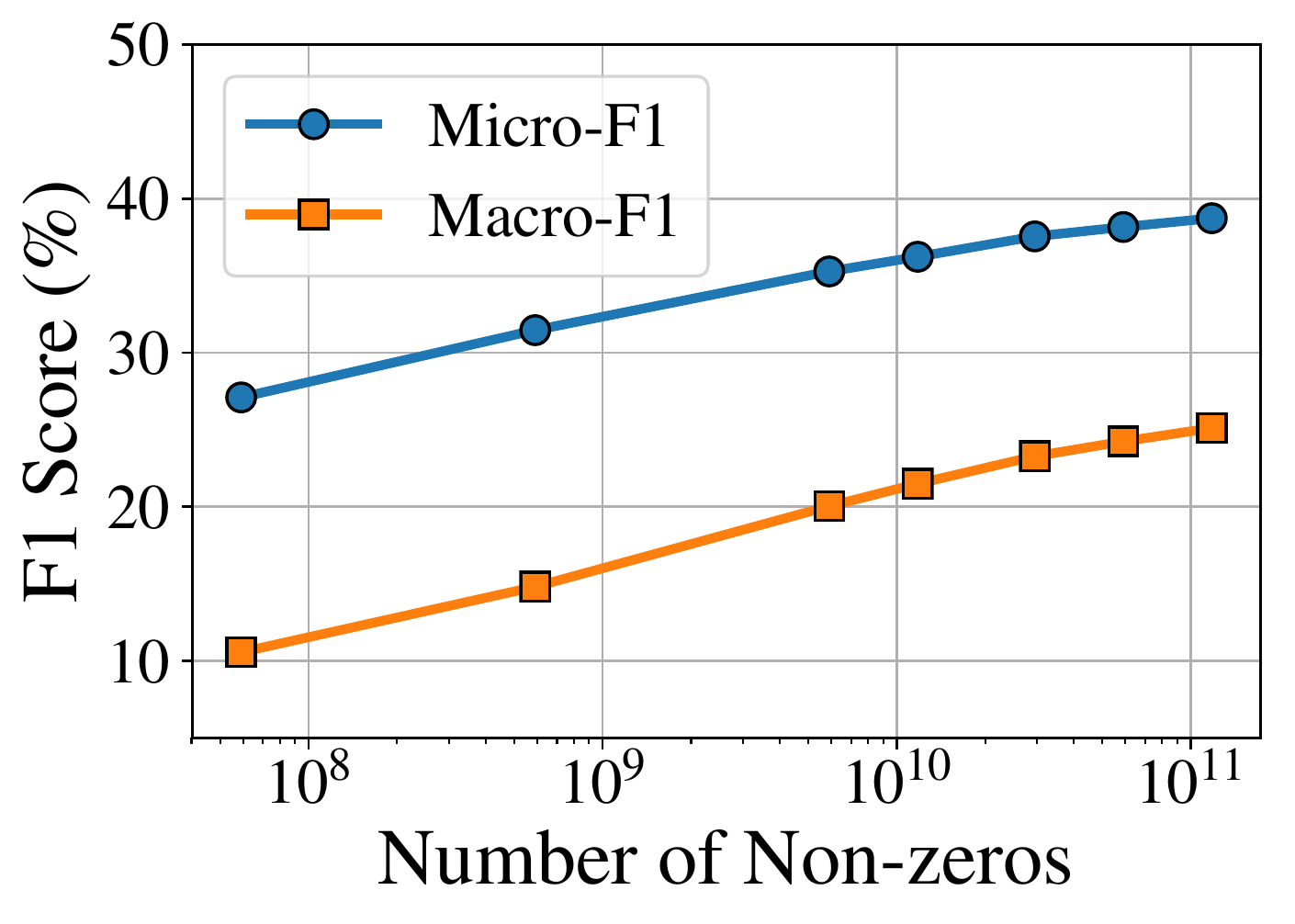}
	    \label{fig:M}
	}
	\subfigure[]{
	    \includegraphics[width=.50\columnwidth]{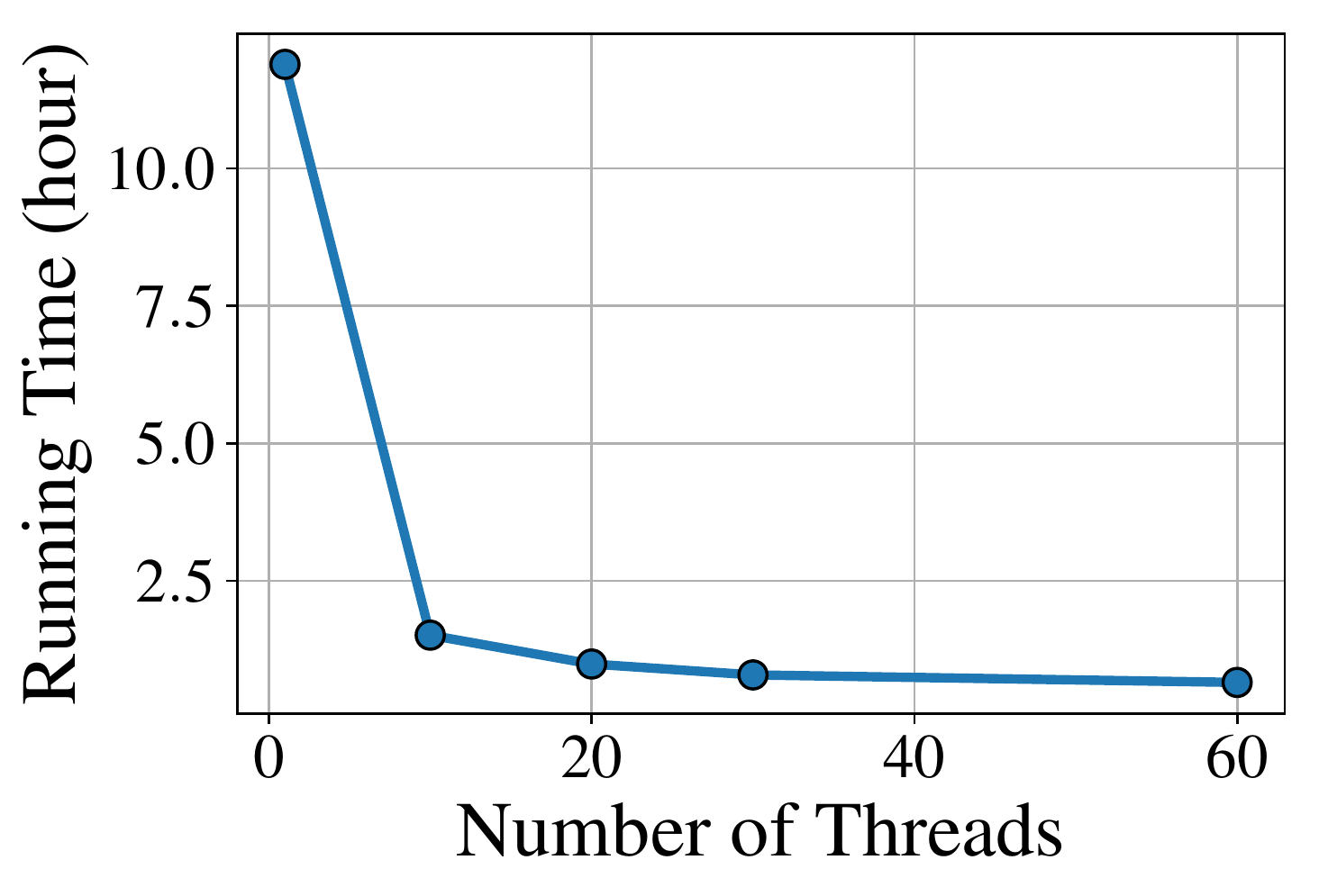}
	    \label{fig:speedup}
	}
	}
	\vspace{-0.1in}
	\caption{Parameter analysis: (a) Prediction performance v.s. embedding dimension $d$; (b) Cattel's Scree Test on singular values. 
	%We plot the top 256 singular values to find a clear drop in the magnitudes or the singular values start to even out. 
	(c) Prediction performance v.s. the number of non-zeros $M$; (d) Running time v.s. the number of threads.}
\end{figure*}

In this section, we discuss how the hyper-parameters influence the performance and efficiency of NetSMF. We report all the parameter analyses on the Flickr dataset with training ratio set to be 10\%.
%\jz{parameters, such as training ratio?}

\vpara{How to Set the Embedding Dimension $d$.} As mentioned in Section~\ref{sec:sparsification}, SVD allows us to determine a ``good'' embedding dimension without supervised information. There are many methods available such as captured energy and Cattell's Scree test~\cite{cattell1966scree}. 
Here we propose to use Cattell's Scree test. Cattell's Scree test plots the singular values and selects a rank $d$ such that there is a clear drop in the magnitudes or 
the singular values start to even out. 
In Flickr, if we sort the singular values in decreasing order, we can observe that the singular values approach 0 when the rank increases to around 100, as shown in \figref{fig:singular}. In our experiments, by varying $d$ form $2^4$ to $2^8$, we reach the best performance at $d=128$, as shown in \figref{fig:dim}, demonstrating the ability of our
matrix factorization based NetSMF for automatically determining the embedding dimension. 

\vpara{The Number of Non-Zeros $M$.}
In theory, 
%we need 
$M$ = $O(Tm\epsilon^{-2}\log{n})$ is required 
to guarantee the approximation error (See Section \ref{sec:sparsification}). 
%However, when we empirically pick a relatively small number of path samples, the prediction performance is still satisfactory. 
Without loss of generality, we empirically set $M$ to be $k\times T \times m$ where $k$ is chosen from 1, 10, 100, 200, 500, 1000, 2000 and investigate how the number of non-zeros influence the quality of learned embeddings.
As shown in \figref{fig:M}, when increasing the number of non-zeros, NetSMF tends to have better prediction performance because the original matrix is being approximated more accurately. On the other hand, although increasing $M$ has a positive effect on the prediction performance, its marginal benefit diminishes gradually. One can observe that setting $M=1000\times T \times m$ (the second-to-the-right data point on each line in \figref{fig:M}) is a good choice that balances NetSMF's efficiency and effectiveness.
%\yx{may need one sentence to conclude this part}

\vpara{The Number of Threads.} In this work, we use a single-machine shared memory implementation with multi-threading acceleration. We report the running time of NetSMF when setting the number of threads to be 1, 10, 20, 30, 60, respectively. As shown in \figref{fig:speedup}, NetSMF takes 12 hours to embed the Flickr network with one thread and 48 minutes to run with 30 threads, achieving a 15$\times$ speedup ratio (with ideal being 30$\times$). %When increasing the number of threads to 60, NetSMF achieves a 18.2$\times$ speedup ratio (with ideal being 60$\times$). 
This relatively good sub-linear speedup supports NetSMF to scale up to very large-scale networks.

\section{Related Work}
\label{sec:related}

In this section, we review the related work of network embedding, large-scale embedding algorithms, and spectral graph sparsification.

    \subsection{Network Embedding}

Network embedding has been extensively studied over the past years~\cite{hamilton2017representation}. 
The success of network embedding has driven a lot of downstream network applications, such as recommendation systems~\cite{ying2018graph}. 
Briefly, recent work about network embedding can be categorized into three genres: (1) Skip-gram based methods that are inspired by word2vec~\cite{mikolov2013efficient}, such as LINE~\cite{tang2015line}, DeepWalk~\cite{perozzi2014deepwalk}, node2vec~\cite{grover2016node2vec}, metapath2vec~\cite{dong2017metapath2vec}, and VERSE~\cite{tsitsulin2018verse}; (2) Deep learning based methods such as \cite{ying2018graph,kipf2017semi};  (3) Matrix factorization based methods such as  GraRep~\cite{cao2015grarep} and NetMF~\cite{qiu2018network}. Among them, NetMF bridges the first and the third categories by unifying a collection of skip-gram based network embedding methods into a matrix factorization framework. In this work, we %inherit the merit of NetMF 
leverage the merit of NetMF and address its limitation in efficiency. 
Among literature, PinSage is notably a network embedding framework for billion-scale networks~\cite{ying2018graph}. 
The difference between NetSMF and PinSage lies in the following aspect. 
The goal of NetSMF is to pre-train general network embeddings in an unsupervised manner, while PinSage is a supervised graph convolutional method with both the objective of recommender systems and existing node features incorporated. That being said, the embeddings learned by NetSMF can be consumed by PinSage for downstream network applications.

\subsection{Large-Scale Embedding Learning}

Studies have attempted to optimize embedding  algorithms for large datasets from different perspectives. Some  focus on improving skip-gram model, while others consider it as  matrix factorization.

\vpara{Distributed Skip-Gram Model.}
Inspired by word2vec~\cite{NIPS2013_5021}, most of the modern embedding learning algorithms are based on the skip-gram model.
%which brings computation issues. For example, its key step, negative sampling, requires to draw samples from a noisy distribution.
There is a sequence of work trying to accelerate the skip-gram model in a distributed system.
For example, \citet{ji2016parallelizing} replicate the embedding matrix on multiple workers and synchronize them periodically; \citet{ordentlich2016network} distribute the columns~(dimensions) of the embedding matrix to multiple executors and synchronize them with a parameter server~\cite{li2014scaling}.
Negative sampling is a key step in skip-gram, which requires to draw samples from a noisy distribution.
\citet{stergiou2017distributed} focus on  the optimization of negative sampling by replacing the roulette wheel selection with a hierarchical sampling algorithm based on the alias method.
More recently, \citet{wang2018billion}
propose a billion-scale network embedding framework by heuristically partitioning the  input graph to small subgraphs, and processing them separately in parallel. However, the performance of their framework highly relies on the quality of graph partition. The drawback for partition-based embedding learning is that the embeddings learned in different subgraphs do not share the same latent space, making it impossible to compare nodes across subgraphs.

\vpara{Efficient Matrix Factorization.}
Factorizing the NetMF matrix, either implicitly~(e.g., LINE~\cite{tang2015line} and DeepWalk~\cite{perozzi2014deepwalk}) or explicitly~(e.g., NetMF~\cite{qiu2018network}), encounters  two issues. 
First, the denseness of this matrix makes computation expensive even for a moderate context window size~(e.g., $T=10$).
Second, the non-linear transformation, i.e., element-wise matrix logarithm, is hard to approximate.
LINE~\cite{tang2015line} solves this problem by setting $T=1$.
With such simplification, it achieves good scalability at the cost of prediction performance. 
NetSMF addresses these issues by efficiently sparsifying the dense NetMF matrix with a theoretically-bounded approximation error.  
%More recently, \citet{zhang2018billion} propose to ignore the element-wise logarithm and factorize a matrix-polynomial directly. However, the non-linearity~(i.e., element-wise logarithm) in embedding learning is proved to have important geometrical meanings~\cite{arora2016latent, hashimoto2016word}. 

%\cite{nie2017unsupervised}

\subsection{Spectral Graph Sparsification}
Spectral graph sparsification has been studied for decades in graph theory~\cite{teng2016scalable}. The task of graph sparsification is to approximate a ``dense'' graph by a ``sparse'' one that can be effectively used in place of the dense one~\cite{teng2016scalable}, which arises in many applications such as scientific computing~\cite{higham2011p}, machine learning~\cite{cheng2015efficient, calandriello2018improved} and data mining~\cite{zhao2015gsparsify}. Our NetSMF model is the first work that incorporates  spectral sparsification algorithms~\cite{cheng2015spectral,cheng2015efficient} into network embedding, which offers a powerful and efficient way to approximate and analyze the random-walk matrix-polynomial in the NetMF matrix.

\section{Conclusion}
\label{sec:conclusion}

In this work, we study network embedding with the goal of achieving both efficiency and effectiveness. 
To address the scalability challenges faced by the NetMF model, we propose to study large-scale network embedding as sparse matrix factorization. 
We present the NetSMF algorithm, which achieves a sparsification of the (dense) NetMF matrix.  
Both the construction and factorization of the sparsified matrix are fast enough to support very large-scale network embedding learning. 
For example, it empowers NetSMF to efficiently embed the Open Academic Graph in 24 hours, whose size is computationally intractable for the dense matrix factorization solution (NetMF). 
Theoretically, the sparsified matrix is spectrally close to the original NetMF matrix with an approximation bound. 
Empirically, our extensive experimental results show that the sparsely learned embeddings by NetSMF are as effective as those from the factorization of the NetMF matrix, leaving it outperform the common network embedding benchmarks---DeepWalk, LINE, and node2vec. 
In other words, among both matrix factorization based methods~(NetMF and NetSMF) and common skip-gram based benchmarks~(DeepWalk, LINE, and node2vec), NetSMF is the only model that achieves both efficiency and performance superiority.

\vpara{Future Work.} NetSMF brings an efficient, effective, and guaranteed solution to network embedding learning. There are multiple tangible research fronts we can pursue.
First, our current single-machine implementation limits the number of samples we can take for large networks. We plan to develop a multi-machine solution in the future to further scale NetSMF. 
Second, building upon NetSMF, we would like to efficiently and accurately learn embeddings for large-scale directed~\cite{cohen2016faster}, dynamic~\cite{kapralov2017single}, and/or heterogeneous networks. 
Third, as the advantage of matrix factorization methods demonstrated, we are also interested in exploring the other matrix definitions that may be effective in capturing different structural properties in networks. 
Last, it would be also interesting to bridge  matrix factorization based network embedding methods with graph convolutional networks.

%We are happy to see its great prospects and potential in industry, including but are not limited to recommendation, user modeling and abnormal detection.

\vpara{Acknowledgements.} We would like to thank Dehua Cheng and Youwei Zhuo from USC for helpful discussions.
Jian Li is supported in part by the National Basic Research Program of China Grant 2015CB358700, 
the National Natural Science Foundation of China Grant 61822203, 61772297, 61632016, 61761146003,
and a grant from Microsoft Research Asia.
Jie Tang is the corresponding author.

\section*{APPENDIX}

We first prove Thm.~\ref{thm:Merror} and Thm.~\ref{thm:log_error} in Section~\ref{sec:error}. The following lemmas will be useful in our proof.

\hide{
\begin{lemma}{(\cite{teng2016scalable})}
Suppose $G$ and $\widetilde{G}$ are $(1+\epsilon)$-spectrally similar. If $\lambda_1\leq \cdots \leq \lambda_{n}$ and $\widetilde{\lambda}_1\leq \cdots \leq \widetilde{\lambda}_{n}$ are the eigenvalues of $\bm{L}$ and $\widetilde{\bm{L}}$, respectively. Then $\forall i \in [n]$,
\beq{
\nonumber
(1-\epsilon) \cdot \widetilde{\lambda}_i \leq \lambda_i \leq (1+\epsilon) \cdot \widetilde{\lambda}_i
}
\end{lemma}
}

\begin{lemma}{(\cite{trefethen1997numerical})}
\label{thm:svd_eig}
Singular values of a real symmetric matrix are the
absolute values of its eigenvalues.
\end{lemma}

\begin{lemma}{(Courant-Fisher Theorem)}
\label{thm:cf}
Let $\bm{A}\in\mathbb{R}^{n\times n}$ be a symmetric matrix with eigenvalues $\lambda_1 \geq \lambda_2 \geq \cdots \geq \lambda_n$, then for $i\in[n]$,
\beq{
\nonumber
\lambda_i = \min_{\dim{(U)}=i} \max_{\bm{x} \in U, \norm{x}{2}=1} \bm{x}^\top \bm{A} \bm{x}. 
}
\end{lemma}

\begin{lemma}
{(\cite{horn_johnson_1991})} 
\label{thm:horn}
Let $\bm{B}, \bm{C}$ be two $n\times n$ symmetric matrices. Then for the decreasingly-ordered singular values $\sigma$ of $\bm{B}, \bm{C}$ and $\bm{BC}$, 
\beq{\nonumber
\sigma_{i+j-1}(\bm{BC}) \leq \sigma_i(\bm{B}) \times \sigma_j(\bm{C}) 
}holds for any $1\leq i, j \leq n$ and $i+j \leq n+1$. 
\hide{
In particular, by fixing $j=1$, 
\beq{
\sigma_{i}(\bm{AB}) \leq \sigma_i(\bm{A}) \times \sigma_1(\bm{B}) 
}for $i=1,\cdots, n$.}
\end{lemma}

\begin{lemma}
\label{thm:Lerror}
Let $\bm{\mathcal{L}} = \bm{D}^{-1/2}\bm{L}\bm{D}^{-1/2}$ and similarly $\widetilde{\bm{\mathcal{L}}} = \bm{D}^{-1/2}\widetilde{\bm{L}}\bm{D}^{-1/2}$. Then all the singular values of $\widetilde{\bm{\mathcal{L}}}-\bm{\mathcal{L}}$ are smaller than $2\epsilon$, i.e., $\forall i\in[n]$,
$\sigma_i(\widetilde{\bm{\mathcal{L}}}-\bm{\mathcal{L}}) < 4\epsilon$.
\end{lemma}
\begin{proof}
 Notice that
\beq{\nonumber
\bm{\mathcal{L}} = \bm{D}^{-1/2}\bm{L}\bm{D}^{-1/2} =\bm{I}  - \sum_{r=1}^T \alpha_r \left( \bm{D}^{-1/2} \bm{A} \bm{D}^{-1/2}\right)^r 
}which is a normalized graph Laplacian whose eigenvalues lie in the interval $[0, 2)$, i.e., for $i\in [n]$, $\lambda_i(\bm{\mathcal{L}}) \in [0, 2)$~\cite{von2007tutorial}.
Since $\widetilde{\bm{L}}$ is a $(1+\epsilon)$-spectral sparsifier of $\bm{L}$, we know that for $\forall \bm{x}\in \mathbb{R}^n$,
\beq{
\nonumber
\frac{1}{1+\epsilon}  \bm{x}^\top \bm{L} \bm{x} \leq \bm{x}^\top \widetilde{\bm{L}} \bm{x} \leq \frac{1}{1-\epsilon}  \bm{x}^\top \bm{L} \bm{x}.
} Let $\bm{x}=\bm{D}^{-1/2}\bm{y}$ which is bijective, we have
\beq{
\nonumber 
\besp{
&\frac{1}{1+\epsilon}  \bm{y}^\top \bm{\mathcal{L}} \bm{y} \leq \bm{y}^\top \widetilde{\bm{\mathcal{L}}} \bm{y} \leq \frac{1}{1-\epsilon}  \bm{y}^\top \bm{\mathcal{L}}\bm{y}\\
\Longrightarrow  &\abs{\bm{y}^\top  (\widetilde{\bm{\mathcal{L}}}-\bm{\mathcal{L}}) \bm{y}} \leq \frac{\epsilon}{1-\epsilon}  \bm{y}^\top \bm{\mathcal{L}}\bm{y} < 2\epsilon\bm{y}^\top \bm{\mathcal{L}}\bm{y} .
}}The last inequality is because we assume $\epsilon < 0.5$. Then, by Courant-Fisher Theorem~(Lemma~\ref{thm:cf}), we can immediately get, $\forall i\in[n]$,
\beq{
\nonumber
\abs{\lambda_i(\widetilde{\bm{\mathcal{L}}}-\bm{\mathcal{L}})} \leq 2\epsilon\lambda_i(\bm{\mathcal{L}}) < 4\epsilon.
}Then, by Lemma~\ref{thm:svd_eig}, $\sigma_i(\widetilde{\bm{\mathcal{L}}}-\bm{\mathcal{L}}) < 4\epsilon, \forall i \in [n]$.
\end{proof}

Given the above lemmas, we can see how the constructed $\widetilde{\bm{M}}$
approximates $\bm{M}$ and how the constructed NetMF matrix sparsifier~(Eq.~\eqref{eq:sparse_deepwalk}) approximates the NetMF matrix~(Eq.~\eqref{eq:deepwalk_matrix}).

\Merror*

\begin{proof} First notice that
$
\widetilde{\bm{M}}-\bm{M} = \bm{D}^{-1}\left(\widetilde{\bm{L}}-\bm{L}\right)\bm{D}^{-1} =\bm{D}^{-1/2} (\widetilde{\bm{\mathcal{L}}}-\bm{\mathcal{L}} ) \bm{D}^{-1/2}
$. Apply Lemma~\ref{thm:horn} twice and use the result from Lemma~\ref{thm:Lerror}, we have
\beq{
\nonumber
\besp{
\sigma_i\left(\widetilde{\bm{M}}-\bm{M}\right) &\leq \sigma_i\left(\bm{D}^{-1/2}\right)\times \sigma_1\left(\widetilde{\bm{\mathcal{L}}}-\bm{\mathcal{L}}\right)\times \sigma_1\left(\bm{D}^{-1/2}\right)\\ 
&\leq  \frac{1}{\sqrt{d_i}} \times 4\epsilon \times \frac{1}{\sqrt{d_{\min}}}   = \frac{4\epsilon}{\sqrt{d_i d_{\min}}}.
}}

\end{proof}

\logerror*

\begin{proof}
It is easy to  observe that $\htln$ is 1-Lipchitz w.r.t. Frobenius norm. So we have
\beq{
\nonumber
\besp{
&\norm{\htln\left(\frac{\vol(G)}{b}\widetilde{\bm{M}}\right)-
\htln\left(\frac{\vol(G)}{b}\bm{M}\right)}{F} \\
 \leq &
 \norm{\frac{\vol(G)}{b}\widetilde{\bm{M}}-
\frac{\vol(G)}{b}\bm{M}}{F} = \frac{\vol(G)}{b} \norm{\widetilde{\bm{M}}-
\bm{M}}{F} \\
=&  \frac{\vol(G)}{b} \sqrt{\sum_{i\in[n]}\sigma^2_i(\widetilde{\bm{M}}-
\bm{M})} 
\leq  \frac{4\epsilon\vol(G)}{b\sqrt{d_{\min}}} \sqrt{\sum_{i=1}^n \frac{1}{d_i}}.
}}
\end{proof}

% https://math.stackexchange.com/questions/433494/upper-bound-of-logarithm

We finally explain the remaining question in Step 1 of NetSMF: After sampling a length-$r$ path $\bm{p}=(u_0, \cdots, u_r)$, why does the algorithm add a new edge to the sparsifier with weight $\frac{rm}{MZ(\bm{p})}$? Our proof relies on two lemmas from \cite{cheng2015spectral}.

\begin{lemma}{(Lemma 3.3 in \cite{cheng2015spectral})}
\label{thm:tau}
Given the path length $r$, the probability for the $\textsf{PathSampling}$ algorithm to sample a path $\bm{p}$ is  $\tau(\bm{p})=\frac{w(\bm{p})Z(\bm{p})}{2rm}$, where $Z(\bm{p})$ is defined in Eq.~\eqref{eq:Z} and 
\beq{
\nonumber
w(\bm{p}) = \frac{\prod_{i=1}^r \bm{A}_{u_{i-1}, u_i}}{\prod_{i=1}^{r-1} \bm{D}_{u_i}}.
}
\end{lemma}

\begin{lemma}{(Theorem 2.2 in \cite{cheng2015spectral})}
\label{thm:weight_origin}
After sampling a length-$r$ path $\bm{p}=(u_0, u_1, \cdots, u_r)$, the weight corresponding to the new edge $(u_0, u_r)$ added to the sparsifier should be $\frac{w(\bm{p})}{\tau(\bm{p})M}$.
\end{lemma}

\begin{theorem}
\label{thm:edge}
After sampling a length-$r$ path $\bm{p}=(u_0, u_1, \cdots, u_r)$ using the \textsf{PathSampling} algorithm~(Alg.~\ref{alg:ps}). The weight of the new edge added to the sparsifier  $\widetilde{\bm{L}}$ is $\frac{2rm}{MZ(\bm{p})}$.
\end{theorem}

\begin{proof}
The proof is to plug the definition of $Z(\bm{p})$, $w(\bm{p})$, and $\tau(\bm{p})$ from Lemma~\ref{thm:tau} into Lemma~\ref{thm:weight_origin}, that is,
\beq{
\nonumber
\frac{w(\bm{p})}{\tau(\bm{p})M} = \frac{w(\bm{p})}{\frac{w(\bm{p})Z(\bm{p})}{2rm} \times M} = \frac{2rm}{MZ(\bm{p})}.
}For unweighted networks, this weight can be simplified to $\frac{m}{M}$, since $Z(\bm{p})=2r$ for unweighted networks.
\end{proof}

\newpage
\bibliographystyle{ACM-Reference-Format}
\balance
\bibliography{main.bib}
\end{document}